\documentclass[conference,compsoc]{IEEEtran}

\ifCLASSOPTIONcompsoc
  \usepackage[nocompress]{cite}
\else
  \usepackage{cite}
\fi

\ifCLASSINFOpdf
  \usepackage[pdftex]{graphicx}
  \graphicspath{{figures/}}
  \DeclareGraphicsExtensions{.jpg,.jpeg,.pdf,.png}
\else
  \usepackage[dvips]{graphicx}
  \graphicspath{{figures/}}
  \DeclareGraphicsExtensions{.eps}
\fi

\usepackage{amsmath}
\usepackage{amssymb}
\usepackage{amsfonts}
\usepackage{amsthm}
\usepackage[ruled,linesnumbered]{algorithm2e}
\usepackage{algpseudocode}

\usepackage{booktabs}
\usepackage{multirow}
\usepackage{array}
\usepackage{adjustbox}
\usepackage{makecell}
\usepackage{colortbl}
\usepackage[table,xcdraw,usenames,dvipsnames]{xcolor}

\usepackage{graphicx}

\usepackage{hyperref}
\usepackage{url}
\hypersetup{
    colorlinks=true,
    linkcolor=red,
    citecolor=cyan,
    filecolor=magenta,      
    urlcolor=magenta,
}

\usepackage[utf8]{inputenc}
\usepackage[T1]{fontenc}
\usepackage{nicefrac}
\usepackage{microtype}
\usepackage{xspace}
\usepackage{enumitem}
\usepackage{pifont}
\usepackage{cleveref}
\usepackage{siunitx}
\usepackage{listings}
\usepackage{marvosym}

\newtheorem{theorem}{Theorem}

\newcommand{\method}{\textsc{SAMark}\xspace}

\definecolor{mygrey}{gray}{0.4}
\usepackage{xcolor}

\begin{document}

\title{\method: A Self-Anchored Text Watermarking with Paragraph-Level Paraphrase Robustness}

\author{
\IEEEauthorblockN{
Jiahao Huo\IEEEauthorrefmark{1}\IEEEauthorrefmark{3},
Wenjie Qu\IEEEauthorrefmark{4},
Yibo Yan\IEEEauthorrefmark{1}\IEEEauthorrefmark{2},
Kening Zheng\IEEEauthorrefmark{3},\\
Jiaheng Zhang\IEEEauthorrefmark{4},
Xuming Hu\IEEEauthorrefmark{1}\IEEEauthorrefmark{2},
Philip S. Yu\IEEEauthorrefmark{3},
Mingxun Zhou\IEEEauthorrefmark{2}\textsuperscript{\Letter}
}
\IEEEauthorblockA{\IEEEauthorrefmark{1}The Hong Kong University of Science and Technology (Guangzhou)}
\IEEEauthorblockA{\IEEEauthorrefmark{2}The Hong Kong University of Science and Technology}
\IEEEauthorblockA{\IEEEauthorrefmark{3}University of Illinois Chicago}
\IEEEauthorblockA{\IEEEauthorrefmark{4}National University of Singapore}
\IEEEauthorblockA{
Emails: jiahaohuotj@gmail.com, mingxunz@ust.hk
}
\IEEEauthorblockA{
$^{\textrm{\Letter}}$ Corresponding author.
}
}

\maketitle

\begin{abstract}
Semantic-level watermarking (SWM) improves robustness against text modifications by treating sentences as the basic unit. However, robustness to paragraph-level paraphrasing remains difficult because such attacks globally disrupt watermark signals by changing sentence order. In this work, we propose \method, a self-anchored watermarking framework that removes the dependency on sentence order by establishing a step-independent green region in semantic space. To improve detectability, we introduce a multi-channel hyperbolic scoring mechanism that amplifies watermark signals while suppressing noise from weakly aligned candidates. We further propose a diversity-aware filtering strategy that combines hard filtering with soft regularization, extending beyond simple n-gram repetition filters to address semantic redundancy. Experimental results show that \method achieves up to 90.2\% TP@FP1\% under typical paragraph-level paraphrasing attacks, outperforming the strongest prior baseline by more than 30\% on average, while maintaining generation quality competitive with unwatermarked text and breaking the robustness-quality trade-off that limits prior methods. Our code will be released at \href{https://github.com/Z1zs/SAMark}{this URL}.

\end{abstract}

\IEEEpeerreviewmaketitle

\section{Introduction} 

The rapid advancement of generative AI (GenAI) has transformed content creation across many domains, from education to software development. Yet the growing ability of these models to produce realistic and persuasive content has also raised concerns about authenticity, attribution, and misuse. LLM watermarking, which embeds imperceptible yet algorithmically detectable signals into model outputs to identify AI-generated text, has emerged as a promising way to mitigate such risks.

During generation, watermarking subtly modifies model outputs while preserving meaning and fluency, embedding a hidden, consistent signal that serves as a verifiable marker of GenAI origin. The watermark can later be detected from the generated output, enabling verification without relying on inherent statistical artifacts of the language model and offering a more resilient way to distinguish machine-generated content from human-authored work. In practical use, AI-generated text is rarely kept unchanged. Users often edit, shorten, or paraphrase generated drafts to match a target style, format, or audience. A watermarking scheme for text traceability therefore needs to preserve its signal under these routine modifications. Currently, most existing watermark methods fall into two categories, namely token-level watermarking and semantic-level watermarking.

\textbf{Token-level Watermarking.} Token-level watermarking (TWM) schemes for LLMs have been widely studied. Green-Red schemes~\cite{kgw,unigram} assign tokens to keyed green and red subsets and bias sampling toward green tokens; the detector then tests whether the generated text contains an unusually high fraction of green tokens. This mechanism provides a simple statistical signal, but it changes the original sampling distribution of the LLM and can degrade text quality~\cite{unbiased}. Distortion-free methods, including Gumbel Watermarking~\cite{exp} and PRC-based pseudorandom error-correcting schemes~\cite{prf}, reduce this distributional distortion and provide theoretical guarantees on text quality. However, their guarantees mainly address generation quality rather than edit robustness. Local modifications such as word deletion, substitution, and paraphrasing can still remove or dilute token-level evidence, which limits the reliability of TWM in adversarial or post-edited settings~\cite{semstamp}.\par
\begin{figure}[t]
\centering
\includegraphics[width=0.98\linewidth]{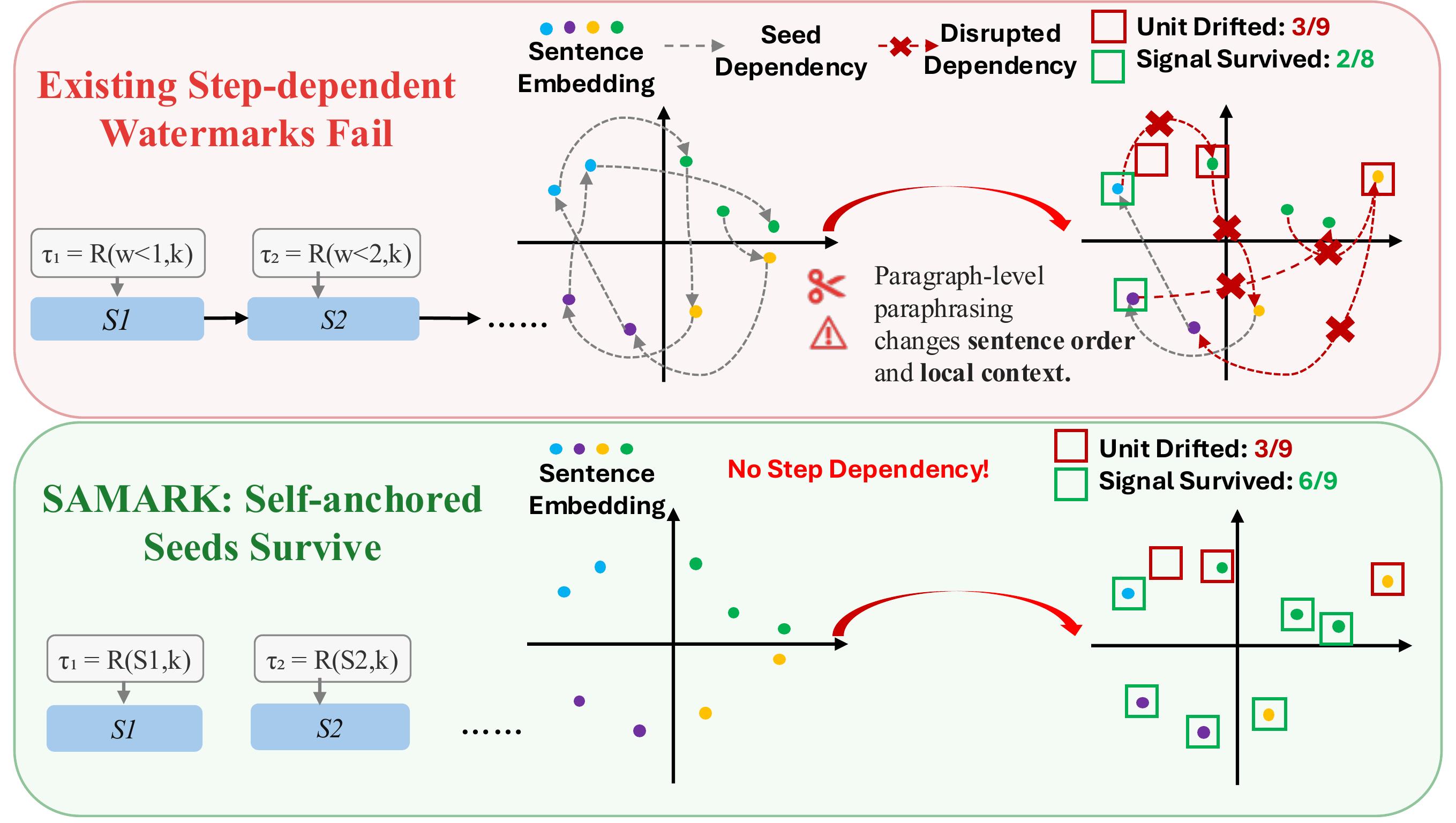}
\caption{Comparison between step-dependent watermarks and \method.}
\label{fig:comp} 
\end{figure}

\textbf{Semantic-level Watermarking.} To improve robustness to such attacks, semantic-level watermarking (SWM) approaches, such as SemStamp~\cite{semstamp,ksemstamp,simmark}, treat a semantically complete sentence as the basic watermarking unit. These methods use rejection sampling to ensure that generated sentences fall within a valid semantic region of the embedding space, analogous to the green list in Green-Red watermarking. Recently, PMark~\cite{huo2026pmark} increased watermark density in SWM by applying multi-channel constraints during sampling, further improving robustness against sentence-level paraphrasing attacks. Despite these efforts, current SWM methods still rely on context hashing or fixed private keys to partition green and red regions during sampling~\cite{semstamp,ksemstamp,huo2026pmark}. This reliance makes them vulnerable to paragraph-level paraphrasing attacks (PPA), which can globally disrupt watermark signals by changing sentence order. Achieving robustness to PPA has long been regarded as an intractable problem in the community~\cite{huo2026pmark}.

\begin{table}[t]
\centering
\caption{Qualitative comparison of representative watermarking methods across the three desired properties. \textbf{Quality} -- AvgRank $\Delta$ vs.\ unwatermarked: \ding{72}\ding{72}\ding{72} ($\Delta\!\le\!-0.5$), \ding{72}\ding{72} ($|\Delta|\!\le\!0.5$), \ding{72} ($\Delta\!>\!0.5$). \textbf{Robustness} -- TP@FP1\% under PPA: \ding{72}\ding{72}\ding{72} ($\ge\!70\%$), \ding{72}\ding{72} ($30\!-\!70\%$), \ding{72} ($<\!30\%$). \textbf{Efficiency} -- sampled tokens per output token relative to KGW; lower is better.}
\label{tab:expected_property}
\begin{tabular}{l|ccc}
\toprule
\textbf{Method} & \textbf{Text Quality} & \textbf{Robustness} & \textbf{Token Cost} \\
\midrule
KGW~\cite{kgw}             & \ding{72}                    & \ding{72}\ding{72}           &  1x \\
SynthID~\cite{synthid}     & \ding{72}\ding{72}\ding{72}  & \ding{72}                    & 1x \\
SemStamp~\cite{semstamp}   & \ding{72}\ding{72}           & \ding{72}\ding{72}           &  64x                  \\
PMark~\cite{huo2026pmark}  & \ding{72}\ding{72}\ding{72}           & \ding{72}\ding{72}           & 100x                  \\
\rowcolor{gray!20}
\method (Ours)             & \ding{72}\ding{72}\ding{72}  & \ding{72}\ding{72}\ding{72}  & 47x         \\
\bottomrule
\end{tabular}
\end{table}
The limitations of TWM and SWM discussed above reflect a broader set of practical requirements that the community has gradually converged on.  Specifically, several desirable properties have emerged. First, the detector for a watermarking scheme should accurately distinguish LLM-generated and human-written content, meaning that it should be unlikely to flag content produced independently of the watermarking keys as watermarked (usually measured by TP@FP1\%). Second, the watermark should not degrade the model's generation quality. In practice, the quality of watermarked text should remain competitive with that of unwatermarked content. Furthermore, watermark detectability should remain robust under adversarial attacks, so a high TP@FP1\% is desirable even for modified or paraphrased outputs. Finally, the watermarking scheme should be efficient and incur minimal latency overhead, enabling practical deployment in LLM services. We summarize a qualitative comparison of representative watermarking methods across the three desired properties in Table~\ref{tab:expected_property}.

\textbf{Is a PPA-robust SWM method possible?} Given the status quo, we ask the following question:

\textit{Can we develop a high-quality semantic-level watermarking scheme that is inherently robust to paragraph-level paraphrasing attacks?} 

An affirmative answer would reshape the practical landscape of SWM. Specifically, our goal is to realize an SWM method that is robust and quality-preserving, thereby enabling practical deployment in LLM services.
\subsection{Our Contribution}

Our first insight is that the vulnerability to paragraph-level paraphrasing attacks of most existing watermarking schemes and their variants~\cite{kgw,synthid,semstamp,huo2026pmark} stems from their reliance on step-dependent random seeds. We show that the intrinsic semantic content of each unit is the most reliable invariant under such attacks, motivating a shift away from step-dependent seed generation.

\begin{figure}[t]
\centering
\includegraphics[width=0.7\linewidth]{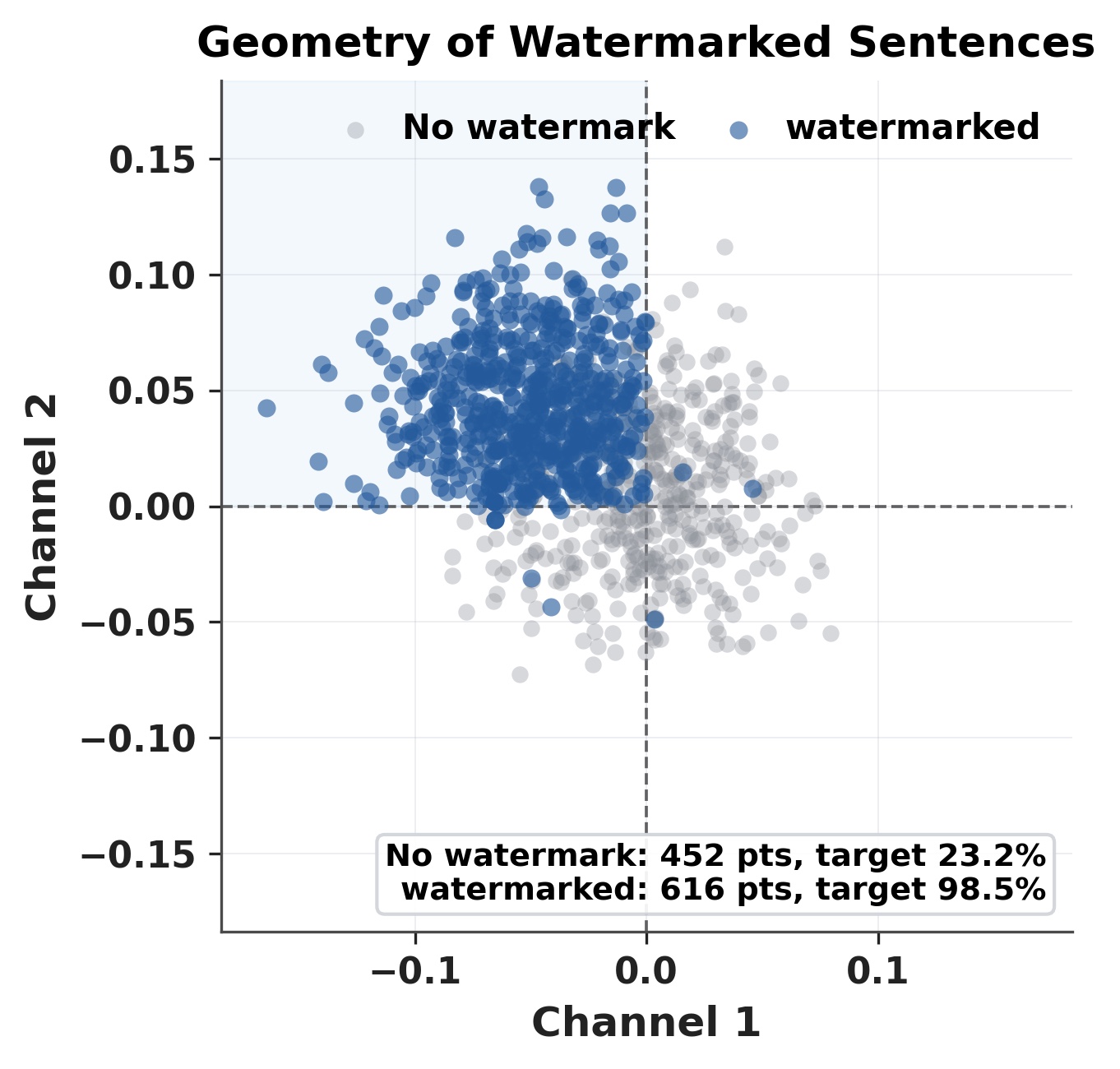}
\caption{Illustration of semantic-level watermark and unwatermarked semantic distribution.}
\label{fig:unw_geo} 
\end{figure}
Our second insight is based on the observation that concentrated watermark scores in semantic-level watermarking can degrade detectability under attacks, as shown in Figure~\ref{fig:unw_geo}. Such concentration increases the risk of sign flips under paraphrasing attacks, posing a direct threat to PPA robustness. We therefore argue that avoiding this concentration near the region boundary is essential for robustness. In practice, we incorporate a hyperbolic sampling strategy into our pipeline to achieve this goal (see Section~\ref{sec:generation}).

\textbf{Self-anchored Paradigm.} Building on these insights, our main technical contribution is \method, a novel self-anchored watermarking framework. The core of our method is to eliminate step dependence by establishing a step-independent green region within semantic space. By making the random seed depend solely on the unit's own semantics, we ensure that the watermark survives both text perturbation and sentence reordering.

\textbf{Multi-channel Hyperbolic Scoring and Quality Control.} To realize this framework, we tackle two technical challenges: signal detectability and generation quality. For signal enhancement, we introduce a multi-channel hyperbolic scoring mechanism that amplifies robust watermark signals while suppressing noise from weakly aligned candidates. To maintain text quality, we propose a diversity-aware filtering strategy that combines n-gram and semantic hard filtering with soft regularization to prevent redundant generation and encourage lexical novelty.

\textbf{Comprehensive Evaluation.} Extensive empirical evaluations demonstrate that \method achieves strong robustness. Under extreme paragraph-level paraphrasing attacks where existing baselines largely fail, our method maintains high detectability. Furthermore, our results show that \method preserves generation quality competitive with that of unwatermarked text, effectively narrowing the traditional robustness-quality trade-off. Ablation studies further confirm the necessity of each component, showing that removing key modules can lead to a 20\% drop in post-attack TP@FP1\%.

The remainder of this paper is organized as follows. Section~\ref{sec:rw} reviews related work on token-level and semantic-level watermarking. Section~\ref{sec:pre} formalizes the problem and analyzes the limitations of current approaches. Section~\ref{sec:method} details the proposed \method framework and its key components. Section~\ref{sec:exp} presents the experimental setup, main results, and comprehensive evaluations. Finally, Section~\ref{sec:dis} discusses possible limitations and future directions, while Section~\ref{sec:con} concludes the paper.

\section{Problem Definition and Preliminaries}\label{sec:pre}

To motivate our approach, we first formalize the general framework of existing watermarking methods and then analyze why they fail under paragraph-level paraphrasing attacks (PPA). This analysis reveals that step-dependent random seeds are the root cause of vulnerability, guiding us toward a self-anchored solution.

\subsection{Problem Setup and General Framework}\label{sec:gen}
\textbf{Problem Setup.} Let $\mathcal{V}$ denote the vocabulary of an LLM $\pi$. The text generation process is denoted by $x \xrightarrow{\pi} y$, where a response $y$ is sampled from the model given an input prompt $x$. We define $\mathcal{V}^*$ as the set of all semantically complete sentences, including the null sentence of length zero. Let $u \in \Omega$ denote the basic unit of watermark generation, where $\Omega$ represents the set of all possible units, such as a token $v \in \mathcal{V}$ in a token-level watermarking (TWM) scheme or a sentence $s \in \mathcal{V}^*$ in a semantic-level watermarking (SWM) scheme. At each generation step $t$, the probability distribution of producing the next unit $u_t \in \Omega$ given the preceding context $u_{<t} = (u_1, \dots, u_{t-1})$ is denoted by $p_\pi(\cdot \mid u_{<t})$.

\textbf{General Framework.} Existing generative watermarking schemes typically consist of three key components: a random seed generator $\mathfrak R$, a sampling algorithm $\mathfrak G$, and a scoring function $\mathfrak S$. We introduce a private key $k$ drawn from a key space $\mathcal{K}$. The watermarked output follows a modified conditional distribution $p_{\pi}^w(u_t \mid u_{<t}; k)$. The unified generation and detection process can be formalized as:
\begin{itemize}
    \item \textbf{Fetching the Random Seed ($\mathfrak R$).} At step $t$, a random seed $r_t$ is generated. For context-hashing-based watermarks, $r_t = \mathfrak R(u_{t-w}, \dots, u_{t-1}; k)$, where $w$ denotes the context window size. For fixed-key watermarks, the seed is step-dependent but context-independent: $r_t = \mathfrak R(t; k)$.
    \item \textbf{Output Sampling ($\mathfrak G$).} The seed $r_t$ defines a green region partition $G_t = \mathcal{P}(\Omega; r_t)$, and the output $u_t$ is preferentially sampled from $G_t$ via $u_t = \mathfrak G(\pi, u_{<t}, G_t) \sim p_{\pi}^w(\cdot \mid u_{<t}; k)$.
    \item \textbf{Detection ($\mathfrak S$).} The detector quantifies the correlation between generated units and reconstructed seeds. The overall detection score is $\mathfrak S(X) = \sum_{t=1}^T \mathfrak S(u_t, r_t)$.
\end{itemize}
Robustness is typically pursued by maximizing $\mathfrak S(X)$ during generation while maintaining text quality through a carefully designed sampling algorithm $\mathfrak G$.

\subsection{Brief Introduction to Existing Methods and Their Limitations}
\textbf{Token- and Semantic-level Green-Red Paradigm.} The Green-Red paradigm was initially introduced for token-level watermarking~\cite{kgw}. At generation step $t$, a seed $r_t=\mathfrak R(u_{t-w},\dots,u_{t-1};k)$ is used to partition the vocabulary $\mathcal V$ into a green list $G_t$ and a red list $R_t=\mathcal V\setminus G_t$, with $|G_t|=\gamma|\mathcal V|$. Given the model logits $\ell_v^{(t)}$ for each token $v\in\mathcal V$, TWM increases the probability of green tokens by adding a positive logit bias $\delta$:
\[
\begin{aligned}
p_\pi^w(v \mid u_{<t};k)
&=
\begin{cases}
\dfrac{\exp(\ell_v^{(t)}+\delta)}{\Lambda}, & v\in G_t, \\
\dfrac{\exp(\ell_v^{(t)})}{\Lambda}, & v\in R_t,
\end{cases}
\\
\text{where } \Lambda
&=
\sum_{v'\in G_t}\exp(\ell_{v'}^{(t)}+\delta)
+
\sum_{v'\in R_t}\exp(\ell_{v'}^{(t)}).
\end{aligned}
\]
Semantic-level watermarking follows the same partitioning principle, but the partition is defined over the unit space $\Omega=\mathcal V^*$ rather than the token vocabulary. Given a context-dependent seed $r_t=\mathfrak R(u_{t-w},\dots,u_{t-1};k)$, the semantic space is partitioned into $G_t=\mathcal P(\Omega;r_t)$ and $R_t=\Omega\setminus G_t$, often through locality-sensitive hashing over sentence embeddings. SWM then performs rejection sampling: a candidate sentence $u$ is accepted only if $u\in G_t$. Equivalently, the watermarked distribution is the model distribution truncated to the green region:
\[
\begin{aligned}
p_\pi^w(u \mid u_{<t};k)
&=
\begin{cases}
\dfrac{p_\pi(u \mid u_{<t})}{Z_t}, & u\in G_t, \\
0, & u\in R_t,
\end{cases}
\\
\text{where } Z_t
&=
\sum_{u'\in G_t}p_\pi(u'\mid u_{<t}).
\end{aligned}
\]

We identify two fundamental vulnerabilities in existing watermarking schemes that stem from their reliance on step-dependent random seeds.

\textit{Context-Dependent Seeds are Fragile.} Most existing TWMs~\cite{kgw, synthid} and SWMs~\cite{semstamp, ksemstamp} rely on context-hashing mechanisms, where the random seed depends strictly on the preceding context: $r_t = \mathfrak R(u_{t-h}, \dots, u_{t-1}; k)$. Paraphrasing attacks involving synonym substitutions and unit modifications perturb the context sequence, thereby altering the reconstructed seeds. The detection score changes to $\mathfrak S(X') = \sum_t \mathfrak S(u'_t, \mathfrak R(u'_{t-h}, \dots, u'_{t-1}; k))$, shifting from the watermarked distribution $\mathcal{D}_w$ back toward the unwatermarked distribution $\mathcal{D}_{n}$ and causing detection failure.

\textit{Step-Dependent Seeds Cannot Survive Reordering.} Some watermarking schemes~\cite{huo2026pmark} use step-dependent fixed private keys $r_t = \mathfrak R(t; k)$ to avoid context dependency. However, they remain vulnerable to unit reordering. When paraphrasing attacks alter the sequence order, the detection score becomes a mismatched sum $\mathfrak S(X') = \sum_t \mathfrak S(u'_t, r_t)$, leading to a drop of up to 40\% in TP@FP1\%.

\textbf{Semantics as the Only Invariant.} From the above analysis, we conclude that \textit{the most reliable invariant property of watermarked text under PPA is the intrinsic semantic content of each unit}, denoted by $\phi(u_t)$. This semantic representation is preserved even when the surface form or step position changes. This insight motivates our self-anchored paradigm: by making the random seed depend solely on the unit's own semantics, i.e., $r_t = \mathfrak R(\phi(u_t); k)$, we can achieve robustness against both perturbation and reordering.

\subsection{The Robustness-Quality Trade-off}
\textit{Pursuing Robustness Degrades Quality.} While existing methods strive to enhance robustness, they often introduce quality degradation~\cite{watermax}. Watermark generators tend to repeat previous n-grams and sentences because of the strong constraints imposed by limited green sets. Although n-gram repetition filters~\cite{markllm} have been used to avoid looped generation, their effectiveness is limited when detecting semantic-level overlap. This trade-off constrains the practical potential of robust watermarking schemes.

\textbf{Smarter Filtering for Better Robustness-Quality Trade-off.} We observe that simple surface-level filters fail to capture semantic redundancy. To address this, we propose a diversity-aware filtering mechanism that operates at both lexical and semantic levels, enabling the selection of candidates that are simultaneously high-scoring (robust) and diverse (high-quality). This pushes the frontier of the robustness-quality trade-off.

\begin{figure*}[t]
\centering
\includegraphics[width=0.98\linewidth]{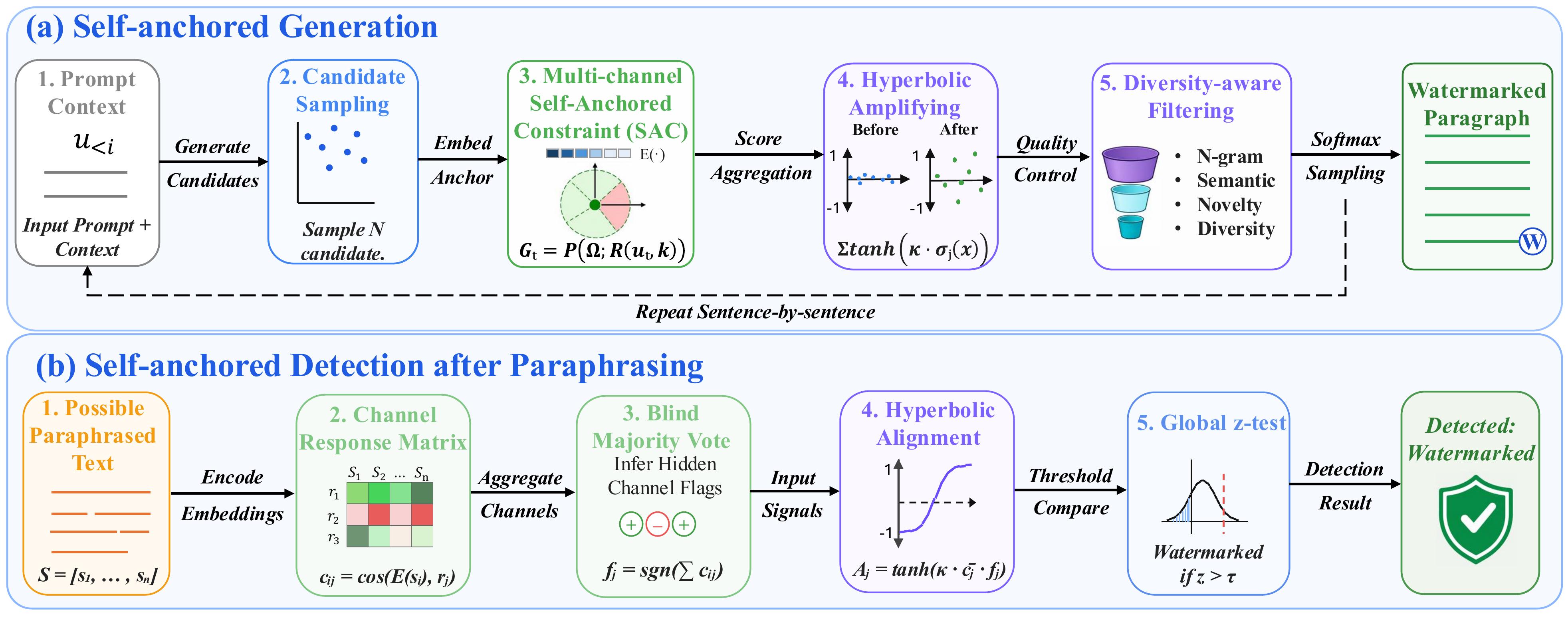}
\caption{Overall framework of \method.}
\label{fig:frame} 
\end{figure*}

\section{SAMark: Self-anchored Watermark}\label{sec:method}
Building on the insights from Section~\ref{sec:pre}, we propose \method, a self-anchored watermarking framework that achieves robustness against paragraph-level paraphrasing by making each sentence's watermark signal depend solely on its own semantics. Our method combines three key components: (1) a theoretically grounded self-anchored constraint that establishes a step-independent green region, (2) a multi-channel hyperbolic scoring mechanism that amplifies watermark signals for improved detectability, and (3) a diversity-aware filtering strategy that maintains text quality without sacrificing robustness. Figure~\ref{fig:frame} illustrates the overall framework of our proposed method.

\subsection{Rethinking Sampling Strategy with Self-Anchored Constraints}
In this section, we formulate the constraints for a self-anchored watermark by replacing the step-dependent components in the general framework with step-independent counterparts. We define the Self-Anchored Condition (SAC) as:
\begin{equation}\label{eq:sac}
\text{SAC:} \quad
\begin{cases}
r_t = \mathfrak R(u_t; k) \\
G_t = \mathcal{P}(\Omega; r_t) \\
u_t = \mathfrak G(\pi, u_{<t}, G_t) \sim p_{\pi}^w(\cdot \mid u_{<t}, k) \in G_t
\end{cases}
\end{equation}

\begin{theorem}[Step-independent Feasible Region]\label{thm:sa_region}
There exists a step-independent feasible region $\mathcal{G} \subseteq \Omega$ such that any unit $u \in \mathcal{G}$ strictly satisfies Eq.~\ref{eq:sac} at any generation step $t$, while any $v \notin \mathcal{G}$ consistently violates it.
\end{theorem}

\begin{proof}
By the definition of the Self-Anchored Condition in Eq.~\ref{eq:sac}, the random seed $r_t = \mathfrak R(u_t; k)$ and the corresponding partition $G_t = \mathcal{P}(\Omega; r_t)$ depend exclusively on the current unit $u_t$ and the fixed key $k$. They are strictly decoupled from the generation step $t$ and the preceding context $u_{<t}$. Therefore, we can explicitly construct the step-independent feasible region $\mathcal{G}$ as:
\begin{equation*}
    \mathcal{G} = \{ u \in \Omega \mid u \in \mathcal{P}(\Omega; \mathfrak R(u; k)) \}.
\end{equation*}
For any unit $u \in \mathcal{G}$, assigning $u_t = u$ at an arbitrary step $t$ deterministically yields the seed $r_t = \mathfrak R(u; k)$ and the partition $G_t = \mathcal{P}(\Omega; r_t)$. By the definition of the set $\mathcal{G}$, it holds that $u_t \in G_t$, which strictly satisfies the constraint in Eq.~\ref{eq:sac}. Conversely, for any unit $v \notin \mathcal{G}$, it inherently follows that $v \notin \mathcal{P}(\Omega; \mathfrak R(v; k))$. Consequently, assigning $u_t = v$ invariably results in $u_t \notin G_t$, which consistently violates Eq.~\ref{eq:sac} regardless of the step $t$. This completes the proof.
\end{proof}
Theorem~\ref{thm:sa_region} demonstrates that the self-anchored conditions in Eq.~\ref{eq:sac} are mathematically equivalent to establishing a globally static green region within the entire space $\Omega$ per response. 

\subsection{\method Generator: Multi-channel Hyperbolic Sampling}\label{sec:generation}

\begin{algorithm}[t]
\small
\SetAlgoLined
\LinesNumbered
\SetKwInOut{Input}{Input}
\SetKwInOut{Output}{Output}
\caption{\method Watermark Generation}
\label{alg:sa_gen}
\Input{LLM $\pi$; context $u_{<1}$; embedding model $\mathcal{E}$; private key $k$; sentences $T$; channels $c$; budget $N$; $\kappa$; $\epsilon$; $\lambda_{\text{div}}, \lambda_{\text{nov}}$, $\theta_{\text{ngram}}, \theta_{\text{sem}}$}
\Output{Generated sentences $u_1,\dots,u_T$}
Derive pivot vectors $V=\{v_1,\dots,v_c\}$ and flags $R=\{r_1,\dots,r_c\}$ from $k$\;
$\mathcal{H} \gets \emptyset$; $\mathcal{V}_{<1} \gets \text{vocab}(u_{<1})$\;
\For{$t \leftarrow 1$ \KwTo $T$}{
  Sample $W=\{x_1,\dots,x_N\}$ with $x_i \sim p_\pi(\cdot \mid u_{<t})$\;
  $W \gets \{x \in W \mid \rho_{\text{ngram}}(x) < \theta_{\text{ngram}}\}$\;
  $W \gets \{x \in W \mid \rho_{\text{sem}}(x) < \theta_{\text{sem}}\}$\;
  $V_{\text{match}} \gets \{x \in W \mid \forall j: \operatorname{sgn}(\cos(\mathcal{E}(x),v_j))= r_j\}$\;
  \ForEach{$x \in V_{\text{match}}$}{
    $S_{\text{wm}}(x) \gets \sum_{j=1}^c \tanh(\kappa \cdot r_j \cdot \cos(\mathcal{E}(x), v_j))$\;
    $S(x) \gets S_{\text{wm}}(x) + \lambda_{\text{div}} D(x) + \lambda_{\text{nov}} \mathcal{N}(x)$\;
  }
  $q(x) \gets \exp(\epsilon \cdot S(x)) / \sum_{x'} \exp(\epsilon \cdot S(x'))$\;
  Sample $u_t \sim q(\cdot)$; $\mathcal{H} \gets \mathcal{H} \cup \{u_t\}$\;
}
\textbf{return} $u_1,\dots,u_T$\;
\end{algorithm}

We adopt the multi-channel sampling paradigm from PMark~\cite{huo2026pmark} and introduce a hyperbolic-enhanced scoring mechanism for candidate sampling. Given $c$ pivot vectors $V = \{v_1, \dots, v_c\}\in \mathbb R^d$ living in the embedding space of text encoder $\mathcal E$ and a target flag pattern $R = \{r_1, \dots, r_c\} \in \{-1, 1\}^c$, the step-independent green region is defined as $\{u \in \Omega \mid \text{sgn}(\langle \mathcal E(u), v_j \rangle) = r_j, \forall j \in \{1, \dots, c\}\}$, with $\langle \cdot,\cdot \rangle$ denoting the inner product and $\text{sgn}(\cdot)$ denoting the signum function. To add feasibility, the target flag pattern $R$ is chosen uniformly at random for each query on the generation side and is not shared with the detection side, reducing storage overhead. The generation process proceeds as follows:

\begin{itemize}
    \item \textbf{Candidate Generation:} Given context $u_{<t}$ and model $\pi$, sample $N$ candidate sentences from the natural model distribution to form the pool $W = \{x_1, \dots, x_N\}$, where $x_i \sim p_\pi(\cdot \mid u_{<t})$.
    
    \item \textbf{Embedding and Channel Matching:} Extract semantic embeddings $\mathcal{E}(x)$ for all $x \in W$. Filter $W$ to retain only candidates whose similarity signs strictly match the target pattern $R$ across all $c$ channels, yielding the matched subset:
    \begin{equation}
        V_{\text{match}} = \{ x \mid \text{sgn}(\langle \mathcal{E}(x), v_j \rangle) = r_j, \forall j \in \{1, \dots, c\} \}.
    \end{equation}
    
    \item \textbf{Hyperbolic Scoring:} For each $x \in V_{\text{match}}$, compute the signed cosine similarity per channel $\sigma_j(x) = r_j \cdot \cos(\mathcal{E}(x), v_j)$. We apply a hyperbolic tangent transform to amplify candidates with strong channel alignment while suppressing near-zero noise:
    \begin{equation}\label{eq:tanh_score}
        S_{\text{wm}}(x) = \sum_{j=1}^c \tanh(\kappa \cdot \sigma_j(x)),
    \end{equation}
    where $\kappa > 0$ is a sharpness parameter controlling the amplification strength (default $\kappa = 30$). The $\tanh$ function saturates at $\pm 1$, ensuring that candidates with large $|\sigma_j(x)|$ contribute maximally while those with weak alignment are suppressed. The final score incorporates diversity regularization (detailed in Section~\ref{sec:diversity}):
    \begin{equation}\label{eq:full_score}
        S(x) = S_{\text{wm}}(x) + \lambda_{\text{div}} \cdot D(x) + \lambda_{\text{nov}} \cdot \mathcal{N}(x),
    \end{equation}
    where $D(x)$ and $\mathcal{N}(x)$ are the semantic diversity and vocabulary novelty terms, respectively.
    
    \item \textbf{Softmax Selection:} Rather than deterministic top-$k$ selection~\cite{watermax}, we employ a softmax distribution over the matched candidates to introduce controlled stochasticity:
    \begin{equation}\label{eq:softmax_select}
        q(x) = \frac{\exp(\epsilon \cdot S(x))}{\sum_{x' \in V_{\text{match}}} \exp(\epsilon \cdot S(x'))},
    \end{equation}
    where $\epsilon > 0$ is a temperature parameter. The watermarked sentence $u_t$ is then sampled according to $u_t \sim q(\cdot)$.
\end{itemize}
\begin{figure}[t]
\centering
\includegraphics[width=0.98\linewidth]{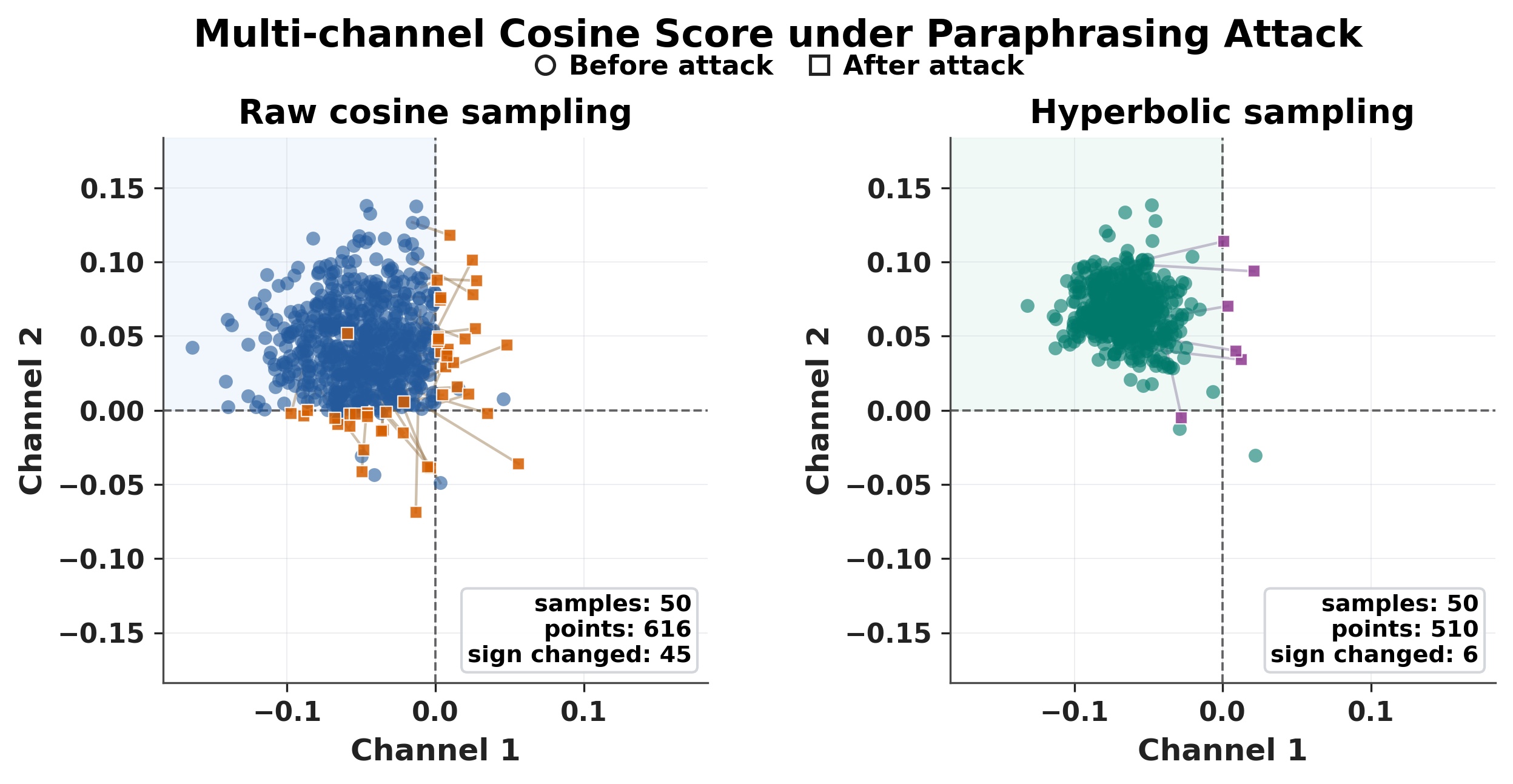}
\caption{Multi-channel semantic deviation under paraphrasing attacks.}
\label{fig:hyper_geo} 
\end{figure}
In Figure~\ref{fig:hyper_geo}, we display the influence of hyperbolic scoring under the Doc-P I attack. Compared with naive cosine scoring in previous SWM methods, hyperbolic scoring prevents the concentration on the boundary of the green regions.
\subsection{\method Detector: Majority-Vote Blind Detection}

\begin{algorithm}[t]
\small
\SetAlgoLined
\LinesNumbered
\SetKwInOut{Input}{Input}
\SetKwInOut{Output}{Output}
\caption{\method Watermark Detection}
\label{alg:sa_detect}
\Input{Text $X$; embedding model $\mathcal{E}$; private key $k$; channels $c$; sharpness $\kappa$; threshold $\tau$}
\Output{\texttt{True} (watermarked) or \texttt{False}}
Segment $X$ into sentences $S=[s_1,\dots,s_n]$\;
Derive pivot vectors $V=\{v_1,\dots,v_c\}$ from $k$\;
\For{$i \leftarrow 1$ \KwTo $n$}{
  \For{$j \leftarrow 1$ \KwTo $c$}{
    $C_{i,j} \gets \cos(\mathcal{E}(s_i), v_j)$\;
  }
}
\For{$j \leftarrow 1$ \KwTo $c$}{
  $\hat{r}_j \gets \operatorname{sgn}\!\left(\sum_{i=1}^n C_{i,j}\right)$\;
}
\For{$i \leftarrow 1$ \KwTo $n$}{
  \For{$j \leftarrow 1$ \KwTo $c$}{
    $A_{i,j} \gets \tanh(\kappa \cdot C_{i,j} \cdot \hat{r}_j)$\;
  }
}
$\bar{A} \gets \frac{1}{nc}\sum_{i,j} A_{i,j}$; $\hat{\sigma}_A \gets \sqrt{\frac{1}{nc}\sum_{i,j} (A_{i,j}-\bar{A})^2}$\;
$z \gets \bar{A} \cdot \sqrt{nc} \,/\, \hat{\sigma}_A$\;
\textbf{return} $(z > \tau)$\;
\end{algorithm}

We introduce a blind, sentence-level detection scheme that extracts watermark evidence across multiple channels without prior knowledge of the target flags $R$. The detection process mirrors the generation-side hyperbolic transform for consistency.

\begin{itemize}
    \item \textbf{Sentence Segmentation:} Segment the input text into a sequence of sentences $S = [s_1, \dots, s_n]$.
    
    \item \textbf{Embedding and Channel Response:} Extract embeddings $\mathcal{E}(s_i)$ and deterministically reconstruct the pivot vectors $V = \{v_1, \dots, v_c\}$ using the shared key. Compute the channel response matrix $C \in \mathbb{R}^{n \times c}$:
    \begin{equation}
        C_{i,j} = \cos(\mathcal{E}(s_i), v_j).
    \end{equation}
    
    \item \textbf{Blind Flag Inference:} Since the detector has no access to the true flags $R$, we infer them via majority voting over the directional responses. For each channel $j$, aggregate the responses across all sentences to estimate the flag as $\hat{r}_j = \text{sgn}\left( \sum_{i=1}^n C_{i,j} \right)$.
    Under the watermarking hypothesis, the majority of sentences will have been selected to align with the true flag pattern, making this inference reliable.
    
    \item \textbf{Score Alignment with Hyperbolic Transform:} Align the raw channel responses using the inferred flags and apply the same hyperbolic transform used during generation:
    \begin{equation}
        A_{i,j} = \tanh(\kappa \cdot C_{i,j} \cdot \hat{r}_j),
    \end{equation}
   where $A_{i,j}$ measures the evidence strength of the $i$-th sentence in the $j$-th channel. Similarly, this transform amplifies strong watermark signals while suppressing noise from weakly aligned sentences.
    
    \item \textbf{Robust Detection via $z$-Test:} Treat the transformed alignment matrix $A$ as a flattened distribution of $n \cdot c$ samples. The global $z$-statistic is computed as:
    \begin{align}
        z = \frac{\bar{A}}{\hat{\sigma}_A / \sqrt{nc}}, \notag 
        \text{where} \quad \bar{A} = \frac{1}{nc}\sum_{i,j} A_{i,j},\\ \quad \hat{\sigma}_A = \sqrt{\frac{1}{nc}\sum_{i,j} (A_{i,j} - \bar{A})^2}.
    \end{align}
    The text is classified as watermarked if $z$ exceeds a predefined confidence threshold $\tau$.
\end{itemize}

\subsection{Diversity-Aware Generation}\label{sec:diversity}
While the watermark scoring mechanism in Eq.~\ref{eq:tanh_score} ensures strong detectability, greedily selecting high-scoring candidates may lead to repetitive or semantically monotonous text~\cite{watermax,markllm}. To maintain generation quality, we introduce a two-stage diversity control mechanism: a hard filtering stage that removes overly similar candidates, followed by soft regularization terms that encourage diverse selections. Such additional optimization ensures the diversity and quality of generated text, avoiding quality degradation resulting from the strong sampling constraints.

\subsubsection{Hard Filtering}
Before scoring, we apply two sequential filters to the candidate pool $W$:

\begin{itemize}
    \item \textbf{N-gram Overlap Filter:} For each candidate $x$, we compute the fraction of its $n$-grams (default $n=4$) that already appear in the generated context $u_{<t}$. Let $\mathcal{G}_n(x)$ denote the set of $n$-grams in $x$ and $\mathcal{G}_n(u_{<t})$ denote those in the context. The overlap ratio is:
    \begin{equation}
        \rho_{\text{ngram}}(x) = \frac{|\mathcal{G}_n(x) \cap \mathcal{G}_n(u_{<t})|}{|\mathcal{G}_n(x)|}.
    \end{equation}
    Candidates with $\rho_{\text{ngram}}(x) \geq \theta_{\text{ngram}}$ are filtered out, where $\theta_{\text{ngram}} \in (0, 1]$ is the overlap threshold.
    
    \item \textbf{Semantic Similarity Filter:} For each remaining candidate, we compute its maximum cosine similarity to any previously selected sentence. Let $\mathcal{H} = \{u_1, \dots, u_{t-1}\}$ be the set of previously generated sentences with embeddings $\{\mathcal{E}(u_i)\}$. The semantic similarity is:
    \begin{equation}
        \rho_{\text{sem}}(x) = \max_{u \in \mathcal{H}} \cos(\mathcal{E}(x), \mathcal{E}(u)).
    \end{equation}
    Candidates with $\rho_{\text{sem}}(x) \geq \theta_{\text{sem}}$ are filtered out, where $\theta_{\text{sem}} \in (0, 1]$ is the similarity threshold.
\end{itemize}

\subsubsection{Soft Regularization}
For candidates passing the hard filters, we augment the watermark score with two diversity-promoting terms:

\begin{itemize}
    \item \textbf{Semantic Diversity Bonus:} We reward candidates that are semantically dissimilar to previously selected sentences:
    \begin{equation}
        D(x) = 1 - \rho_{\text{sem}}(x) = 1 - \max_{u \in \mathcal{H}} \cos(\mathcal{E}(x), \mathcal{E}(u)).
    \end{equation}
    This term ranges in $[0, 1]$, with higher values indicating greater semantic novelty.
    
    \item \textbf{Vocabulary Novelty Bonus:} We encourage candidates that introduce new vocabulary. Let $\mathcal{V}(x)$ denote the set of unique words in $x$ and $\mathcal{V}_{<t}$ denote the accumulated vocabulary from the context. The novelty score is:
    \begin{equation}
        \mathcal{N}(x) = \frac{|\mathcal{V}(x) \setminus \mathcal{V}_{<t}|}{|\mathcal{V}(x)|},
    \end{equation}
    representing the fraction of words in $x$ that have not appeared in the generated text so far.
\end{itemize}

The final selection score in Eq.~\ref{eq:full_score} balances watermark strength with diversity through the weights $\lambda_{\text{div}}$ and $\lambda_{\text{nov}}$. Setting these weights to zero recovers the pure watermark-driven selection, while positive values encourage more varied and natural-sounding text.

\section{Experimental Results}\label{sec:exp}
In this section, we conduct extensive experiments to answer the following research questions:
(\textbf{RQ1}) How does \method perform under various PPAs?
(\textbf{RQ2}) Can \method preserve generation quality in practice?
(\textbf{RQ3}) What is the computational cost of \method compared with other SWM baselines?
(\textbf{RQ4}) How sensitive is \method to its key components or parameters?

\subsection{Experiment Setup}
\textbf{Dataset and Baselines.} Following previous work~\cite{semstamp,huo2026pmark}, we evaluate \method on 500 samples from C4~\cite{c4} and BOOKSUM~\cite{booksum}, using Llama-3.1-8B~\cite{llama3} and Mistral-Small-3.1-24B-Base-2503~\cite{mistral2025small31} as backbone models. Our baselines follow the official MarkLLM implementation~\cite{markllm}, including token-level methods KGW~\cite{kgw}, MorphMark~\cite{morphmark}, EXP~\cite{exp}, and SynthID~\cite{synthid}, as well as semantic-level baselines SemStamp~\cite{semstamp}, k-SemStamp~\cite{ksemstamp}, and PMark~\cite{huo2026pmark}. For \method, we set the number of channels to $2$ and the sample budget to $N=64$. We report two variants: \textbf{\method}, which uses the full diversity-aware sampling and soft-regularization pipeline, and \textbf{\method$^\dagger$}, which disables these optimization terms. More implementation details are provided in Appendix~\ref{app:setup}.

\textbf{Metrics.} Following previous work, we assess watermark effectiveness in terms of detectability through \textbf{TP@FP1\%}, \textbf{TP@FP5\%}, \textbf{AUC}, and robustness under paragraph-level attacks, such as Paraphrase Attack \textbf{Doc-P} and Back-Translation Attack \textbf{Doc-T} by GPT-4.1-mini~\cite{openai2022chatgpt}. We use two paraphrasing prompts: \textit{Doc-P I}, ``Please rewrite the following text:'' and \textit{Doc-P II}, ``Please rewrite the following text, avoiding using the same words or phrases in the original text:''. We regard \textit{Doc-P II} as the most representative and challenging scenario in practice. Appendix~\ref{app:attack} provides more details.

Although perplexity (PPL) is widely used in prior work, it is insufficient for our setting: low-quality outputs with repeated sentences can receive artificially low PPL because repetitive patterns are easy for language models to predict. We therefore include three diversity-oriented reference metrics: \textbf{Sentence Duplicate percentage (SD)}, which is lower when generation is less redundant; \textbf{Distinct-2 (D-2)}, which is higher when lexical variety is greater; and \textbf{4-gram Repeat percentage (4g)}, which is lower when text is less repetitive. To evaluate text quality fairly, we use GPT-4.1-mini to conduct pairwise blind review between each watermarking method and the unwatermarked baseline for each sample, producing \textbf{Win (PS-W)}, \textbf{Lose (PS-L)}, and \textbf{Tie (PS-T)}. Exact computation details are given in Appendix~\ref{app:quality_metrics}.

\subsection{Detectability, Robustness and Text Quality (\textbf{RQ1 \& RQ2})}\label{sec:main_res}
To answer \textbf{RQ1} and \textbf{RQ2}, we comprehensively compare \method with seven widely used token-level watermarks and three semantic-level watermarks under various attacks. Tables \ref{tab:booksum_main} and \ref{tab:c4_main} present results on BOOKSUM and C4 datasets respectively. Full results with additional metrics are provided in Appendix~\ref{app:full_results}.

\begin{table}[t]
    \centering
    \caption{Results on \textbf{BOOKSUM}. We report TP@FP1\%, pairwise LLM Judge W/L/T\% and AvgRank ($\Delta$ vs.\ unwatermarked). \textbf{Bold} = best, \underline{underlined} = second-best. $^\dagger$ denotes \method without quality optimization.}
    \setlength{\tabcolsep}{1pt}
    \renewcommand{\arraystretch}{0.95}
    \resizebox{\columnwidth}{!}{
    \begin{tabular}{l|cccc|ccc|c}
    \toprule
    \multirow{2}{*}{\textbf{Method}}
    & \multicolumn{4}{c|}{\textit{TP@FP1\%}}
    & \multicolumn{3}{c|}{\textit{Quality}}
    & \textit{Rank} \\
    \cline{2-9}
    & \textbf{No Atk} & \textbf{Doc-P I} & \textbf{\textit{Doc-P II}} & \textbf{Doc-T}
    & \textbf{PS-W} & \textbf{PS-L} & \textbf{PS-T}
    & \textbf{Avg}~($\Delta$) \\
    \midrule
    \multicolumn{9}{c}{\textit{Mistral-Small-3.1-24B-Base-2503}} \\
    \midrule
    None & -- & -- & -- & -- & -- & -- & -- & 5.76~(+0.00) \\
    EXP & 99.8 & 38.2 & 17.2 & 29.6 & 52.4 & 46.2 & 1.4 & 5.56~(-0.20) \\
    KGW & 100.0 & \underline{64.0} & 34.4 & 54.4 & 38.0 & 50.2 & 11.8 & 6.65~(+0.89) \\
    UPV & 99.0 & 61.4 & 23.6 & 38.0 & 47.2 & 51.4 & 1.4 & 6.12~(+0.35) \\
    SynthID & 99.8 & 23.8 & 10.4 & 12.6 & 54.2 & 38.6 & 7.2 & 5.11~(-0.65) \\
    MorphMark & 99.4 & 56.0 & 34.2 & 45.8 & 47.6 & 51.4 & 1.0 & 6.08~(+0.32) \\
    SemStamp & 96.0 & 49.8 & 39.8 & 58.6 & 53.2 & 39.2 & 7.6 & 5.04~(-0.72) \\
    $k$-SemStamp & 98.6 & 61.7 & 43.1 & 65.9 & 51.4 & 47.2 & 1.4 & 4.94~(-0.82) \\
    PMark & 98.0 & 51.5 & \underline{57.5} & \underline{77.0} & 57.2 & 37.6 & 5.2 & 4.61~(-1.15) \\
    \rowcolor{gray!30}
    \method & 98.0 & \textbf{90.5} & \textbf{88.1} & \textbf{87.1} & 47.9 & 45.2 & 6.9 & 5.14~(-0.62) \\
    \rowcolor{gray!15}
    \method$^\dagger$ & 98.7 & 88.7 & 89.0 & 90.3 & 33.0 & 59.0 & 8.0 & {--} \\
    \midrule
    \multicolumn{9}{c}{\textit{Nemotron-Nano-9B-v2-Base}} \\
    \midrule
    None & -- & -- & -- & -- & -- & -- & -- & 5.06~(+0.00) \\
    EXP & 99.8 & 42.2 & 16.2 & 37.8 & 30.6 & 45.8 & 23.6 & 5.81~(+0.75) \\
    KGW & 100.0 & \underline{76.8} & 45.6 & \underline{66.8} & 32.0 & 42.4 & 25.6 & 5.97~(+0.91) \\
    UPV & 99.4 & 71.0 & \underline{48.8} & 61.6 & 39.2 & 38.2 & 22.6 & 5.42~(+0.36) \\
    SynthID & 99.8 & 31.4 & 15.6 & 21.2 & 47.8 & 33.4 & 18.8 & 4.38~(-0.67) \\
    MorphMark & 100.0 & 74.0 & 46.2 & 64.0 & 36.4 & 40.4 & 23.2 & 5.57~(+0.51) \\
    SemStamp & 97.0 & 46.9 & 34.5 & 64.3 & 39.2 & 41.6 & 19.2 & 5.24~(+0.18) \\
    $k$-SemStamp & 97.6 & 51.7 & 37.1 & 63.7 & 35.0 & 45.0 & 20.0 & 5.58~(+0.52) \\
    PMark & 96.6 & 24.2 & 28.3 & 46.3 & 26.0 & 49.8 & 24.2 & 6.22~(+1.16) \\
    \rowcolor{gray!30}
    \method & 98.8 & \textbf{92.6} & \textbf{90.2} & \textbf{88.2} & 34.2 & 46.0 & 19.8 & 5.75~(+0.69) \\
    \rowcolor{gray!15}
    \method$^\dagger$ & 99.4 & 86.6 & 84.8 & 91.2 & 20.4 & 56.4 & 23.2 & {--} \\
    \bottomrule
    \end{tabular}
    }
    \label{tab:booksum_main}
    \end{table}
    \begin{table}[t]
    \centering
    \caption{Results on \textbf{C4}. We report TP@FP1\% for detection, pairwise LLM Judge W/L/T\% and AvgRank for text quality. \textbf{Bold} = best, \underline{underlined} = second-best. $^\dagger$ denotes \method without quality optimization.}
    \setlength{\tabcolsep}{1pt}
    \renewcommand{\arraystretch}{0.95}
    \resizebox{\columnwidth}{!}{
    \begin{tabular}{l|cccc|ccc|c}
    \toprule
    \multirow{2}{*}{\textbf{Method}}
    & \multicolumn{4}{c|}{\textit{TP@FP1\%}}
    & \multicolumn{3}{c|}{\textit{Quality}}
    & \textit{Rank} \\
    \cline{2-9}
    & \textbf{No Atk} & \textbf{Doc-P I} & \textbf{\textit{Doc-P II}} & \textbf{Doc-T}
    & \textbf{PS-W} & \textbf{PS-L} & \textbf{PS-T}
    & \textbf{Avg}~($\Delta$) \\
    \midrule
    \multicolumn{9}{c}{\textit{Mistral-Small-3.1-24B-Base-2503}} \\
    \midrule
    None & -- & -- & -- & -- & -- & -- & -- & 5.69~(+0.00) \\
    EXP & 99.4 & 24.6 & 11.5 & 36.4 & 50.8 & 49.0 & 0.2 & 5.38~(-0.32) \\
    KGW & 100.0 & \underline{77.7} & 29.0 & 64.2 & 41.6 & 57.8 & 0.6 & 6.56~(+0.87) \\
    UPV & 99.7 & 67.1 & 44.5 & \underline{71.5} & 44.0 & 55.6 & 0.4 & 6.17~(+0.48) \\
    SynthID & 99.2 & 44.1 & 25.1 & 44.1 & 59.1 & 40.5 & 0.4 & 4.64~(-1.05) \\
    MorphMark & 99.7 & 55.5 & 28.0 & 52.7 & 41.7 & 57.9 & 0.4 & 6.17~(+0.48) \\
    SemStamp & 95.1 & 28.6 & 17.2 & 54.7 & 56.4 & 43.6 & 0.0 & 5.19~(-0.50) \\
    $k$-SemStamp & 92.5 & 22.5 & 15.5 & 44.5 & 58.4 & 41.4 & 0.2 & 4.93~(-0.77) \\
    PMark & 96.2 & 40.0 & \underline{48.9} & 67.6 & 64.8 & 35.0 & 0.2 & 4.17~(-1.52) \\
    \rowcolor{gray!30}
    \method & 94.5 & \textbf{78.6} & \textbf{76.8} & \textbf{84.2} & 82.0 & 5.0 & 13.0 & 6.10~(+0.41) \\
    \rowcolor{gray!15}
    \method$^\dagger$ & 95.2 & 78.8 & 76.9 & 90.0 & 78.0 & 6.0 & 16.0 & {--} \\
    \midrule
    \multicolumn{9}{c}{\textit{Nemotron-Nano-9B-v2-Base}} \\
    \midrule
    None & -- & -- & -- & -- & -- & -- & -- & 5.23~(+0.00) \\
    EXP & 100.0 & 39.1 & 16.9 & 58.6 & 36.7 & 58.4 & 5.0 & 6.21~(+0.98) \\
    KGW & 100.0 & 55.9 & 27.3 & \underline{56.2} & 39.6 & 56.0 & 4.4 & 5.98~(+0.75) \\
    UPV & 99.7 & \underline{75.0} & \underline{53.2} & \underline{79.5} & 42.2 & 53.2 & 4.6 & 5.64~(+0.41) \\
    SynthID & 100.0 & 42.9 & 22.6 & 53.0 & 54.8 & 41.4 & 3.8 & 4.60~(-0.63) \\
    MorphMark & 99.8 & 48.3 & 20.4 & 49.0 & 45.0 & 50.8 & 4.2 & 5.76~(+0.53) \\
    SemStamp & 96.4 & 40.4 & 27.9 & 61.7 & 52.4 & 44.8 & 2.8 & 4.73~(-0.50) \\
    $k$-SemStamp & 90.5 & 15.6 & 14.6 & 32.1 & 39.2 & 56.2 & 4.6 & 5.56~(+0.33) \\
    PMark & 96.4 & 24.7 & 28.3 & 45.0 & 40.0 & 55.0 & 5.0 & 5.31~(+0.08) \\
    \rowcolor{gray!30}
    \method & 93.8 & \textbf{78.3} & \textbf{72.9} & \textbf{82.3} & 69.4 & 5.2 & 25.4 & 5.98~(+0.75) \\
    \rowcolor{gray!15}
    \method$^\dagger$ & 94.8 & 76.6 & 73.8 & 83.8 & 28.0 & 68.6 & 3.4 & {--} \\
    \bottomrule
    \end{tabular}
    }
    \label{tab:c4_main}
    \end{table}

\textbf{\method achieves dominant robustness under most attack scenarios.} \method attains the highest detection rates across attack settings on both models and datasets. On Mistral with BOOKSUM under Doc-P II, \method achieves 88.1\% TP@FP1\%, compared with 57.5\% for the next-best baseline, yielding a 30.6-point improvement. Similar advantages persist on Nemotron and C4. While UPV occasionally shows higher TP@FP5\% in specific cases, AUC margins confirm \method's stronger detection stability. Most baselines fail to maintain detectable signals under extreme paraphrasing, whereas \method preserves robust detection.

\textbf{The robustness advantage expands with stronger adversaries.} \method's performance gap over baselines widens as attack intensity increases. From Doc-P I to Doc-P II on Mistral-BOOKSUM, KGW drops by 29.6 points in TP@FP1\%, while \method decreases by only 2.4 points. SynthID drops from 23.8\% to 10.4\%, whereas \method remains nearly constant. The same pattern appears across models: on Nemotron-C4, UPV decreases by 21.8 points between attack levels, while \method drops by only 5.4 points. This stability under escalating adversarial pressure indicates that semantic channel encoding preserves watermark signals even when surface-level features are substantially altered.

\begin{figure}[t] 
\centering
\includegraphics[width=\linewidth]{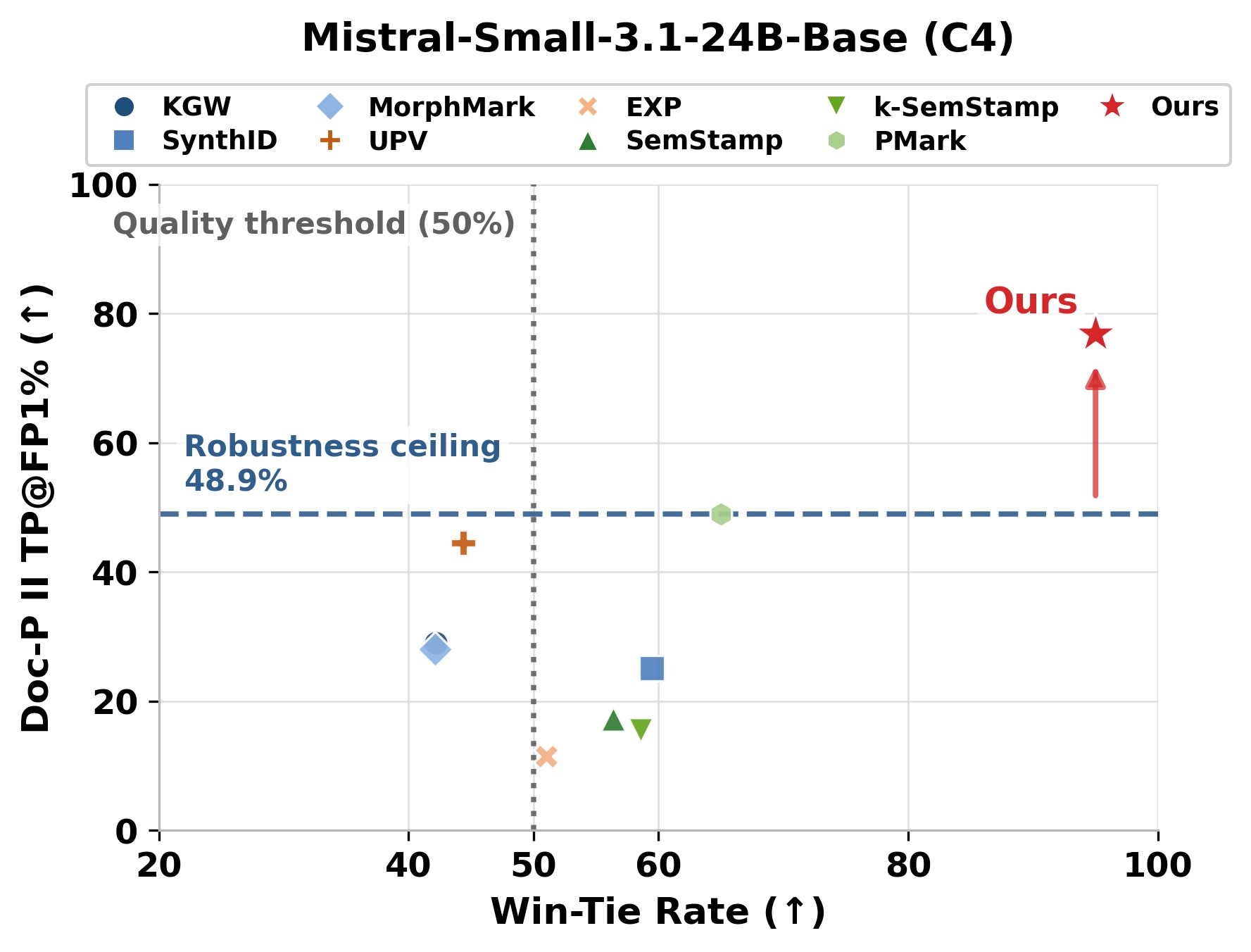}
\caption{Trade-off between robustness and text quality on Mistral-Small-3.1-24B-Base with C4. The x-axis reports the PS-W + PS-T (no worse than) rate.}
\label{fig:trade-off} 
\end{figure}

\textbf{\method maintains competitive text quality while achieving dominant robustness.} Pairwise LLM-judge results show that \method achieves high win rates on C4 and outperforms all baselines, suggesting that sentence-level filtering operation can introduce beneficial diversity. As shown in Figure~\ref{fig:trade-off}, \method surpasses the robustness ceiling of existing methods while maintaining text quality no worse than the unwatermarked baseline. The AvgRank columns in Tables~\ref{tab:booksum_main} and~\ref{tab:c4_main}, obtained via GPT-4.1-mini blind ranking across all methods (including the unwatermarked baseline), support the same conclusion. On Mistral+BOOKSUM, \method achieves an AvgRank of 5.14 ($\Delta=-0.62$), outperforming the unwatermarked baseline (5.76) and remaining comparable to SynthID (5.11). KGW consistently ranks worst across settings ($\Delta$ from +0.75 to +0.91), confirming that hard token-level constraints degrade fluency. Among semantic-level methods, \method ranks comparably to or better than SemStamp and $k$-SemStamp in most settings while delivering substantially higher robustness. On Nemotron, most watermarking methods rank above the unwatermarked baseline, suggesting that watermark-induced sampling diversity can occasionally improve generation quality. Overall, both pairwise and ranking evaluations indicate that \method's robustness gains do not come at the expense of text quality.

\textbf{\method shows architecture-agnostic and distribution-agnostic consistency.} \method maintains TP@FP1\% above 72\% under all attack scenarios across both backbone models, while baselines show high variance. Under Doc-P II, PMark drops by 29.2 points across models, whereas \method exhibits minimal fluctuation across datasets with different distributions. \method consistently achieves at least 97.0\% AUC across all 16 attack-model-dataset combinations, and no baseline matches this uniformity. These results suggest that the self-anchored encoding strategy captures fundamental linguistic invariants rather than model-specific or domain-specific artifacts.
\begin{figure}[t]
\centering
\includegraphics[width=\linewidth]{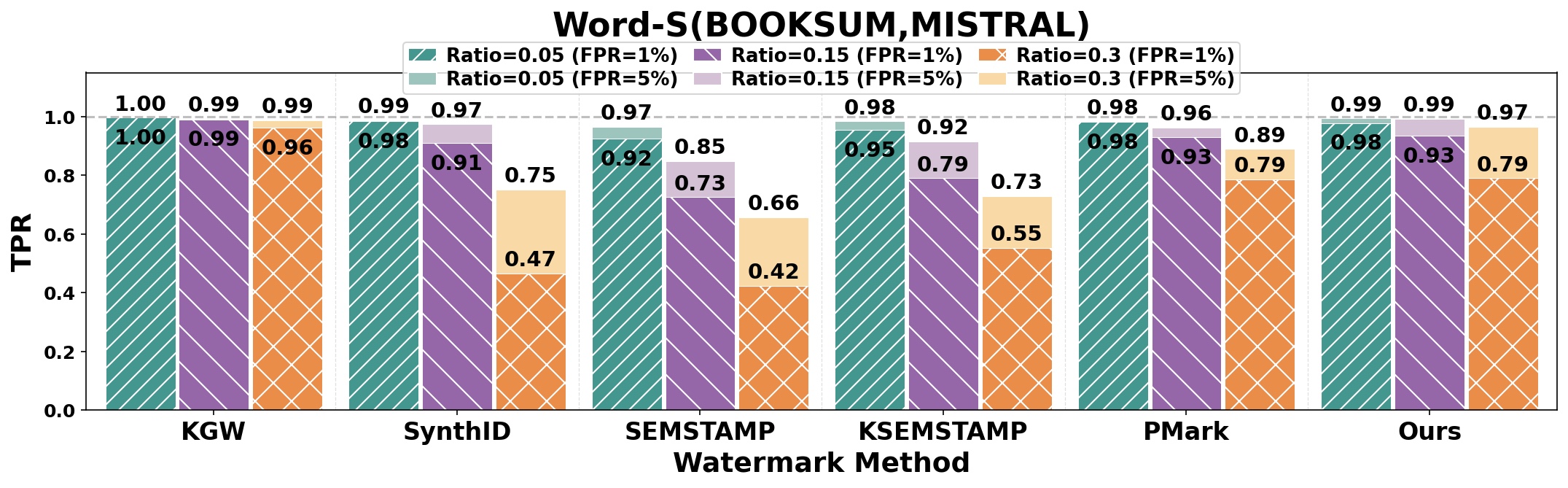}
\caption{Word-S attack results.}
\label{fig:word_s}
\end{figure}

\textbf{Robustness under word-level attacks.} Beyond paragraph-level paraphrasing, we further evaluate word-level perturbations that alter local lexical surface forms. Specifically, Word-S replaces a fraction of words with synonyms or context-compatible alternatives, while preserving most sentence-level semantics. As shown in Figure~\ref{fig:word_s}, existing semantic-level watermarks such as SemStamp, $k$-SemStamp, and PMark degrade substantially because their verification relies on stable sentence segmentation and sentence-internal structure. A similar trend appears under Word-D (Appendix~\ref{app:word_attack}). Token-level methods also show a robustness--quality tension, since methods with better fluency tend to lose detectability faster as perturbation strength increases, while methods with stronger robustness incur larger quality loss. In contrast, \method remains more stable across perturbation types, because its signal is anchored to sentence-level semantics rather than exact token realization. This result is consistent with Tables~\ref{tab:booksum_main} and~\ref{tab:c4_main}, and shows that \method mitigates the structural fragility of prior semantic-level baselines without inheriting the strongest quality penalties of robust token-level methods.

\subsection{Efficiency Performance \textbf{(RQ3)}}
\begin{table}[htbp]
\centering
\caption{Computation overhead of different methods. *: PMark introduces additional overhead for reconstructing the green-red threshold at detection. Abbrev.: Lat = latency, Tok = token, Sen = output sentence, Out = output token.}
\resizebox{\columnwidth}{!}{
\begin{tabular}{l|cccc|cc}
\hline
\multirow{2}{*}{Method} & \multicolumn{4}{c|}{Generation} & \multicolumn{2}{c}{Detection} \\
\cline{2-7}
 & Lat/Tok & Lat/Sen & Tok/Out & Tok/Sen & Lat/Tok & Tok/Sen \\
\hline
KGW       & 0.037 & 0.91 & 1.00 & 24.52 & 0.0008 & 0.00 \\
SynthID   & 0.039 & 0.81 & 1.00 & 20.78 & 0.0007 & 0.00 \\
MorphMark & 0.037 & 0.87 & 1.00 & 23.31 & 0.0009 & 0.00 \\
UPV       & 0.041 & 0.85 & 1.00 & 20.47 & 0.0001 & 0.00 \\
EXP       & 0.279 & 5.71 & 1.00 & 21.06 & 0.0025 & 0.00 \\
\hline
SemStamp  & 0.395 & 7.58 & 64.50 & 1294.76 & 0.0063 & 0.00 \\
k-SemStamp & 0.361 & 6.65 & 65.05 & 1255.72 & 0.0043 & 0.00 \\
PMark     & 0.205 & 4.63 & 49.73 & 1146.01 & 0.1861 & $1068.86^*$ \\
\method       & 0.218 & 4.61 & 47.44 & 1080.97 & 0.0014 & 0.00 \\
\hline
\end{tabular}
}
\label{tab:overhead}
\end{table}
\textbf{\method achieves the lowest overall computation overhead among semantic-level methods.} We compare the computation overhead of \method and existing baselines on BOOKSUM with Mistral-Small-3.1-24B-Base-2503. Following the default MarkLLM configuration, token-level baselines are evaluated with the \texttt{transformers} backend, whereas semantic-level methods are evaluated with the \texttt{vLLM} backend. As shown in Table~\ref{tab:overhead}, token-level methods generally have lower generation latency because they modify token logits in a single decoding pass and do not perform sentence-level candidate sampling. This lower latency, however, comes with the robustness limitations observed in Tables~\ref{tab:booksum_main} and~\ref{tab:c4_main}.

Within the semantic-level group, \method achieves a generation latency of 0.22s per token and 4.61s per sentence, comparable to PMark and approximately $1.7\times$ faster than SemStamp and $k$-SemStamp. This efficiency stems from the self-anchored design, which avoids the iterative context-hashing overhead required by SemStamp-family methods. Moreover, \method consumes the fewest sampled tokens per output token (47.44) and per sentence (1080.97), indicating that the diversity-aware filtering effectively reduces redundant sampling attempts. On the detection side, \method achieves a latency of only 0.0014s per token with negligible additional token consumption, as detection requires only a single forward pass through the embedding model followed by lightweight algebraic operations (majority voting and $z$-test). In contrast, PMark incurs 0.1861s per token and 1068.86 sampled tokens per sentence due to LLM-based re-generation, making \method's detection over $130\times$ faster.

\subsection{Hyperparameter Sensitivity Analysis \textbf{(RQ4)}}

To answer \textbf{RQ4}, we conduct a systematic sensitivity analysis on the key hyperparameters of \method using Mistral on BOOKSUM. We examine the four diversity-related parameters, namely the semantic diversity weight $\lambda_{\text{div}}$, the vocabulary novelty weight $\lambda_{\text{nov}}$, the semantic similarity threshold $\theta_{\text{sem}}$, and the n-gram overlap threshold $\theta_{\text{ngram}}$, as well as the sample budget $N$. For each parameter, we sweep over a range of values while holding the others at their defaults.

\begin{table}[t]
\centering
\caption{Ablation study on diversity-aware filtering. Each row disables one component respectively.}
\resizebox{\columnwidth}{!}{
\begin{tabular}{l|ccc|ccc}
\toprule
\multirow{2}{*}{\textbf{Config}}
& \multicolumn{3}{c|}{\textit{Quality}}
& \multicolumn{3}{c}{\textit{LLM Judge}} \\
\cline{2-7}
& \textbf{SD\%}$\downarrow$ & \textbf{D-2}$\uparrow$ & \textbf{4g\%}$\downarrow$
& \textbf{PS-W} & \textbf{PS-L} & \textbf{PS-T} \\
\midrule
Default & 0.27 & 0.825 & 4.96 & 47.9 & 45.2 & 6.9 \\
w/o $\theta_{\text{ngram}}$ & 1.35 & 0.792 & 8.31 & 48.8 & 47.5 & 3.8 \\
w/o $\theta_{\text{sem}}$ & 1.25 & 0.797 & 7.55 & 47.5 & 52.5 & 0.0 \\
w/o $\lambda_{\text{div}}$ & 0.10 & 0.776 & 7.26 & 45.0 & 53.8 & 1.3 \\
w/o $\lambda_{\text{nov}}$ & 0.71 & 0.762 & 8.70 & 46.3 & 47.5 & 6.3 \\
All Off & 7.65 & 0.570 & 29.11 & 33.0 & 59.0 & 8.0 \\
\bottomrule
\end{tabular}
}
\label{tab:ablation}
\end{table}

\textbf{Each diversity component contributes to text quality.} Table~\ref{tab:ablation} presents the ablation results where each diversity component is individually disabled. Removing the n-gram overlap filter raises the 4-gram repetition rate from 4.96\% to 8.31\% and the sentence duplicate rate from 0.27\% to 1.35\%. Disabling the semantic similarity filter produces a similar degradation, with the 4-gram repetition rate increasing to 7.55\% and the sentence duplicate rate rising to 1.25\%. Turning off the diversity weight $\lambda_{\text{div}}$ reduces Distinct-2 from 0.825 to 0.776 and lowers the LLM judge win rate from 47.9\% to 45.0\%, while disabling the novelty weight $\lambda_{\text{nov}}$ causes the largest drop in Distinct-2 to 0.762 and the highest 4-gram repetition rate among single-component ablations at 8.70\%. When all four components are disabled simultaneously, the quality degradation compounds dramatically: the sentence duplicate rate surges to 7.65\%, the 4-gram repetition rate reaches 29.11\%, and the LLM judge win rate drops to 33.0\%. These results confirm that the hard filtering and soft regularization mechanisms are complementary, each addressing a distinct aspect of text diversity.

\begin{table}[t]
\centering
\caption{Sensitivity to diversity parameters on Mistral+BOOKSUM. Default values are marked with $\ast$.}
\resizebox{\columnwidth}{!}{
\begin{tabular}{c|cccc||c|cccc}
\toprule
$\lambda_{\text{div}}$ & \textbf{SD\%}$\downarrow$ & \textbf{D-2}$\uparrow$ & \textbf{4g\%}$\downarrow$ & \textbf{W\%}
& $\lambda_{\text{nov}}$ & \textbf{SD\%}$\downarrow$ & \textbf{D-2}$\uparrow$ & \textbf{4g\%}$\downarrow$ & \textbf{W\%} \\
\midrule
0.0 & 0.10 & 0.78 & 7.26 & 45.0
& 0.0 & 0.71 & 0.76 & 8.70 & 46.3 \\
0.1 & 0.45 & 0.78 & 7.67 & 47.5
& 0.1 & 0.61 & 0.78 & 7.10 & 48.8 \\
0.35$^\ast$ & 0.27 & 0.82 & 4.96 & 47.9
& 0.2$^\ast$ & 0.27 & 0.82 & 4.96 & 47.9 \\
0.7 & 0.83 & 0.83 & 5.16 & 57.5
& 0.3 & 0.31 & 0.84 & 4.37 & 51.3 \\
1.0 & 0.62 & 0.83 & 5.40 & 60.0
& 0.5 & 0.10 & 0.85 & 3.27 & 56.3 \\
\midrule
$\theta_{\text{sem}}$ & \textbf{SD\%}$\downarrow$ & \textbf{D-2}$\uparrow$ & \textbf{4g\%}$\downarrow$ & \textbf{W\%}
& $\theta_{\text{ngram}}$ & \textbf{SD\%}$\downarrow$ & \textbf{D-2}$\uparrow$ & \textbf{4g\%}$\downarrow$ & \textbf{W\%} \\
\midrule
0.6 & 0.71 & 0.84 & 4.73 & 46.3
& 0.2 & 0.32 & 0.83 & 3.58 & 51.3 \\
0.7 & 0.10 & 0.84 & 3.35 & 48.8
& 0.3 & 0.31 & 0.81 & 4.81 & 45.0 \\
0.8$^\ast$ & 0.27 & 0.82 & 4.96 & 47.9
& 0.4$^\ast$ & 0.27 & 0.82 & 4.96 & 47.9 \\
0.9 & 0.52 & 0.80 & 7.51 & 48.8
& 0.6 & 0.52 & 0.81 & 6.80 & 40.0 \\
1.0 & 1.25 & 0.80 & 7.55 & 47.5
& 1.0 & 1.35 & 0.79 & 8.31 & 48.8 \\
\bottomrule
\end{tabular}
}
\label{tab:sensitivity}
\end{table}

\textbf{Soft regularization weights exhibit a clear quality improvement trend.} As shown in Table~\ref{tab:sensitivity}, increasing $\lambda_{\text{div}}$ from 0 to 1.0 steadily improves Distinct-2 from 0.78 to 0.83 and raises the LLM judge win rate from 45.0\% to 60.0\%, confirming that the semantic diversity bonus effectively encourages varied sentence selection. Similarly, increasing $\lambda_{\text{nov}}$ from 0 to 0.5 reduces the 4-gram repetition rate from 8.70\% to 3.27\% and improves Distinct-2 from 0.76 to 0.85, demonstrating that the vocabulary novelty bonus successfully promotes lexical variety. Both parameters show monotonic quality gains without degrading detection performance, as the watermark scoring mechanism in Eq.~\ref{eq:tanh_score} remains the dominant component in the final selection score.

\textbf{Hard filtering thresholds control the quality floor.} The semantic similarity threshold $\theta_{\text{sem}}$ and the n-gram overlap threshold $\theta_{\text{ngram}}$ act as hard constraints that remove overly similar candidates before scoring. Relaxing $\theta_{\text{sem}}$ from 0.7 to 1.0 increases the 4-gram repetition rate from 3.35\% to 7.55\% and raises the sentence duplicate rate from 0.10\% to 1.25\%, as more semantically redundant candidates enter the scoring pool. A similar pattern holds for $\theta_{\text{ngram}}$: increasing it from 0.2 to 1.0 raises the 4-gram repetition rate from 3.58\% to 8.31\%. These results confirm that hard filtering provides a reliable quality floor by eliminating degenerate candidates early in the pipeline.

\begin{table}[t]
\centering
\caption{Effect of sample budget $N$ on Mistral+BOOKSUM.}
\resizebox{\columnwidth}{!}{
\begin{tabular}{c|ccc|ccc}
\toprule
\multirow{2}{*}{$N$}
& \multicolumn{3}{c|}{\textit{TP@FP1\%}}
& \multicolumn{3}{c}{\textit{Quality}} \\
\cline{2-7}
& \textbf{Doc-P I} & \textbf{Doc-P II} & \textbf{Doc-T}
& \textbf{SD\%}$\downarrow$ & \textbf{D-2}$\uparrow$ & \textbf{4g\%}$\downarrow$ \\
\midrule
64 & 93.1 & 88.8 & 90.6 & 0.27 & 0.82 & 4.96 \\
128 & 100.0 & 93.8 & 93.8 & 0.52 & 0.88 & 1.53 \\
200 & 97.0 & 95.0 & 90.5 & 0.08 & 0.84 & 4.49 \\
\bottomrule
\end{tabular}
}
\label{tab:budget}
\end{table}

\textbf{Larger sample budgets improve both robustness and quality.} Table~\ref{tab:budget} shows that increasing the sample budget $N$ from 64 to 128 and 200 consistently improves detection rates under all attack scenarios. At $N=128$, \method achieves 100.0\% TP@FP1\% under Doc-P I and 93.8\% under Doc-P II, compared to 93.1\% and 88.8\% at $N=64$. The quality metrics also benefit from larger budgets, as a richer candidate pool provides more opportunities to find sentences that simultaneously satisfy the watermark constraint and the diversity criteria. The 4-gram repetition rate drops from 4.96\% at $N=64$ to 1.53\% at $N=128$, and Distinct-2 improves from 0.82 to 0.88. We adopt $N=64$ as the default to balance performance with computational cost, as discussed in Section~\ref{sec:main_res}.

\begin{table}[t]
\centering
\caption{Ablation on the hyperbolic transform. ``Gen'' and ``Det'' indicate whether $\tanh$ scoring is applied during generation and detection respectively.}
\resizebox{\columnwidth}{!}{
\begin{tabular}{cc|cccc}
\toprule
\textbf{Gen} & \textbf{Det}
& \textbf{No Atk} & \textbf{Doc-P I} & \textbf{Doc-P II} & \textbf{Doc-T} \\
\midrule
\ding{55} & \ding{55} & 94.0 & 70.6 & 67.4 & 49.0 \\
\ding{55} & \ding{51} & 95.2 & 75.2 & 68.8 & 48.8 \\
\ding{51} & \ding{55} & 99.8 & 94.8 & 93.2 & 95.2 \\
\ding{51} & \ding{51} & 99.4 & 94.8 & 92.8 & 94.0 \\
\bottomrule
\end{tabular}
}
\label{tab:transform}
\end{table}

\textbf{The hyperbolic transform during generation is the dominant factor for robustness.} Table~\ref{tab:transform} isolates the contribution of the $\tanh$ scoring mechanism in generation and detection on Mistral+BOOKSUM. Without $\tanh$ in generation, detection rates under Doc-P II and Doc-T remain below 70\%, regardless of whether $\tanh$ is applied during detection. Enabling $\tanh$ in generation alone boosts Doc-P II from 67.4\% to 93.2\% and Doc-T from 49.0\% to 95.2\%, a gain of more than 25 points. This improvement arises because the generation-side $\tanh$ transform actively selects candidates with stronger channel alignment, producing sentences whose watermark signals are more resilient to paraphrasing. Adding $\tanh$ during detection on top of $\tanh$ generation yields comparable results, confirming that the generation-side transform is the primary contributor. We retain $\tanh$ in both stages for pipeline consistency.

\subsection{Case Study}
\begin{figure*}[t]
    \centering

    \begin{minipage}[t]{0.32\textwidth}
        \centering
        \includegraphics[width=\linewidth]{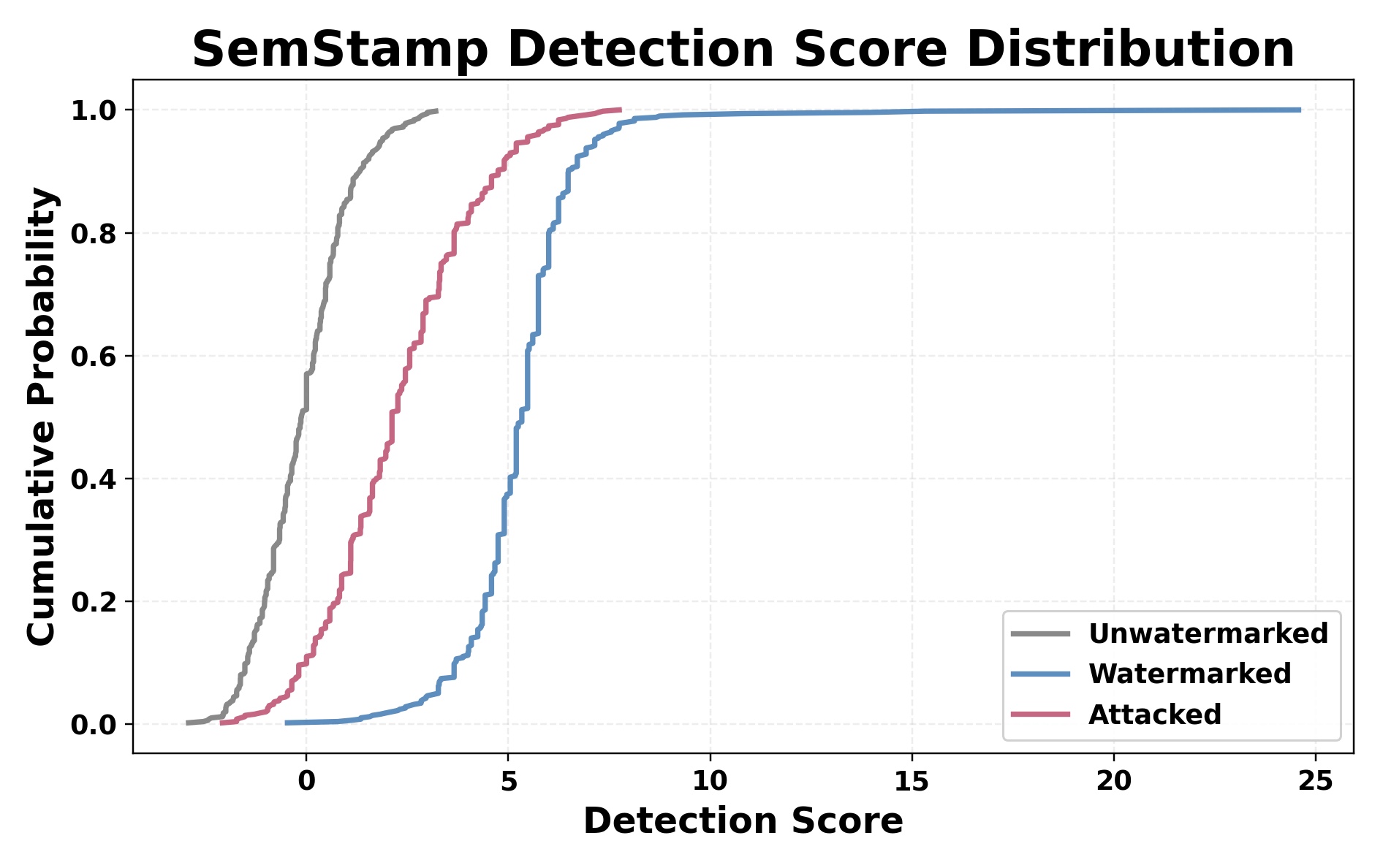}
        \vspace{-0.5em}
    \end{minipage}
    \hfill
    \begin{minipage}[t]{0.32\textwidth}
        \centering
        \includegraphics[width=\linewidth]{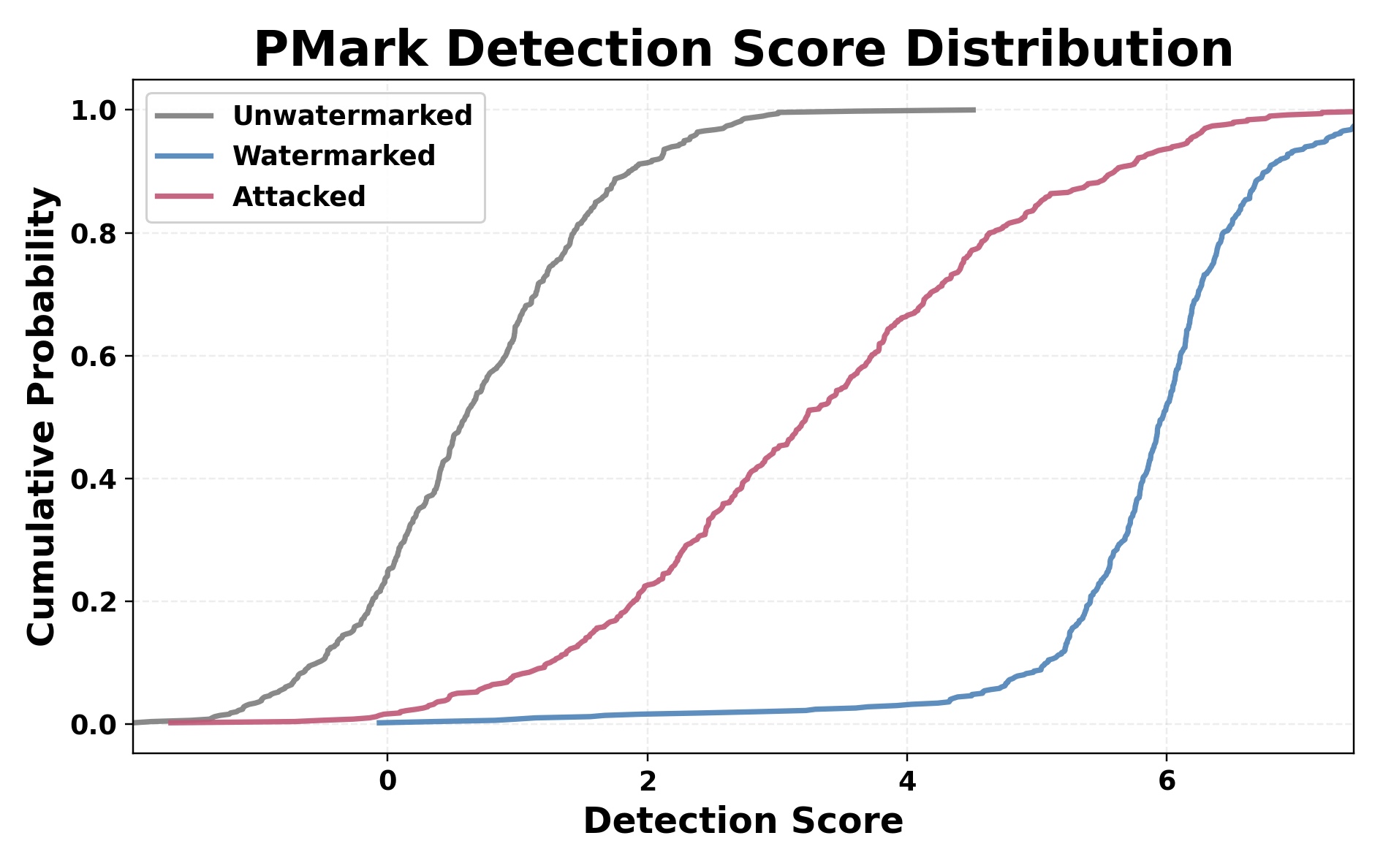}
        \vspace{-0.5em}
    \end{minipage}
    \hfill
    \begin{minipage}[t]{0.32\textwidth}
        \centering
        \includegraphics[width=\linewidth]{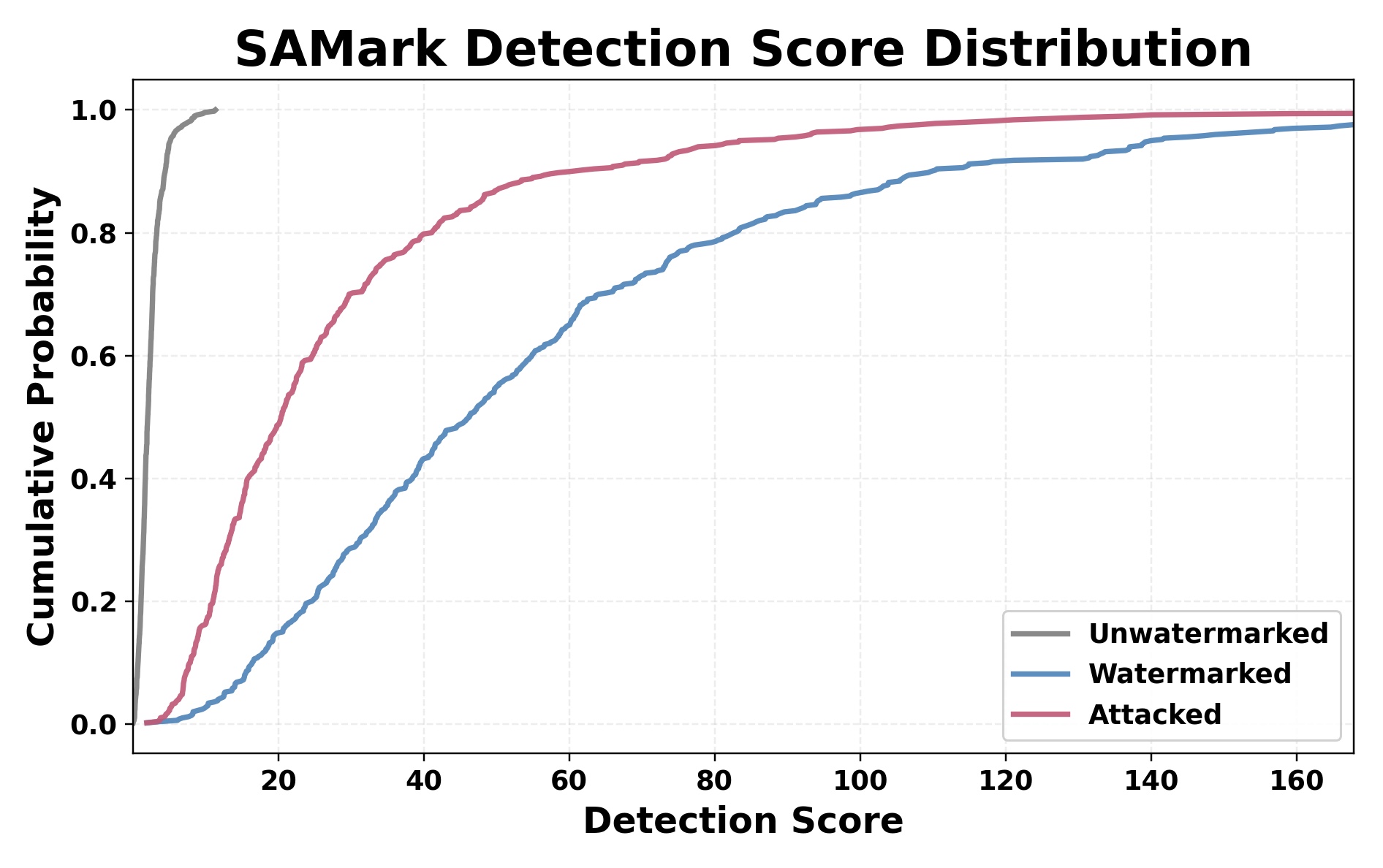}
        \vspace{-0.5em}
    \end{minipage}

    \caption{
    Detection score distributions for unwatermarked, watermarked, and attacked samples on BOOKSUM with Mistral-Small-3.1-24B-Base-2503.
    }
    \label{fig:score_dist}
\end{figure*}
\begin{figure}[t]
    \centering
    \includegraphics[width=\linewidth]{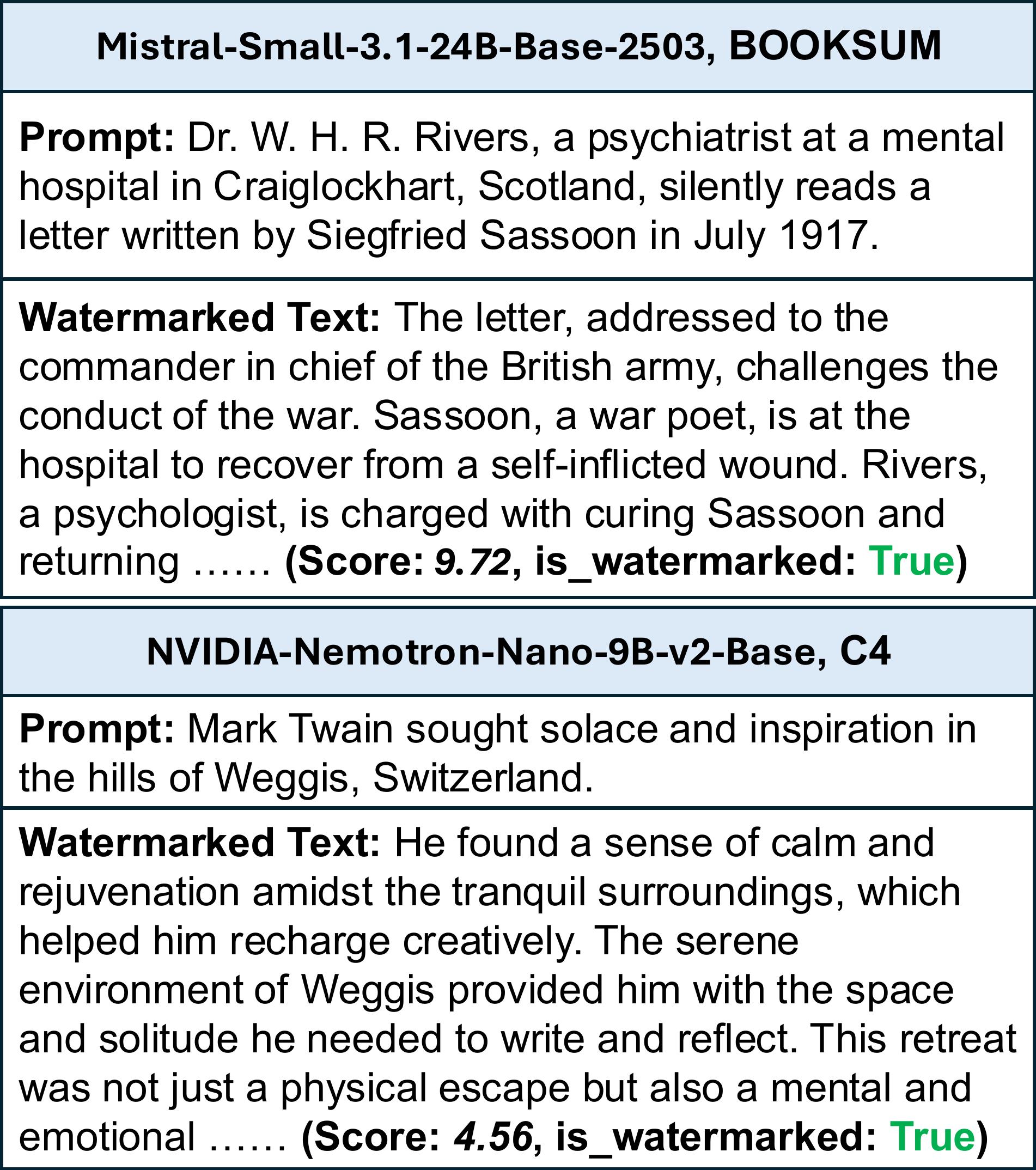}
    \caption{
    Case study of text generated by \method.
    }
    \label{fig:case_study}
    \vspace{-6mm}
\end{figure}

\textbf{Score distribution analysis.}
Figures~\ref{fig:score_dist} show the same trend at the distribution level. For SemStamp and PMark, the attacked score distribution shifts toward the unwatermarked distribution, which increases overlap and weakens threshold-based detection. For \method, attacked samples still remain clearly separated from unwatermarked samples, although their scores are lower than the original watermarked scores. This separation explains the higher robustness observed in the main results, where \method preserves a usable detection margin under aggressive paraphrasing. We also include certain specific cases in Figure~\ref{fig:case_study} for reference.
\section{Related Work}\label{sec:rw}
In recent years, LLM watermarking techniques have advanced rapidly, with KGW~\cite{kgw} being a seminal contribution that embeds statistically detectable signals into text by introducing logit biases to pseudorandomly selected tokens. Following the taxonomy adopted in prior work, existing approaches can be broadly categorized into three types: (1) zero-bit watermarking, (2) multi-bit watermarking, and (3) semantic-level watermarking. Our work falls within the third category, while drawing inspiration from the other two.

\subsection{Zero-bit Watermarking}
The primary objective of zero-bit watermarking is to determine whether a given text is machine-generated, with KGW~\cite{kgw} being the most representative. Subsequent work optimizes the KGW framework in both \emph{detectability} and \emph{text quality}.

\noindent\textbf{Detectability.} EWD~\cite{lu2024entropy} enhances sensitivity in low-entropy contexts (e.g., code) by prioritizing high-entropy tokens via entropy-based $z$-score weighting. WinMax~\cite{kirchenbauer2023reliability} combats signal dilution in mixed texts using a sliding-window approach to dynamically compute multi-window $z$-scores and select maxima, boosting robustness. ITS and EXP~\cite{kuditipudi2024robust} further enhance detectability by embedding a secret key sequence into tokens and computing alignment costs during detection to mitigate adversarial token manipulations.

\noindent\textbf{Text Quality.} Existing methods mitigate quality degradation by either embedding watermarks in high-entropy, low-semantic-impact positions using entropy-aware red-green list distributions~\cite{wouters2023optimizing}, optimizing list distributions via semantic alignment to disperse similar words across lists~\cite{chen2024watme}, or increasing context-relevant word selection probability~\cite{fu2024watermarking}. Beyond autoregressive LLMs, recent studies have also investigated order-agnostic watermarking for non-autoregressive paradigms~\cite{chen2025watermark,wu2025dmark}.

\subsection{Multi-bit Watermarking}
Many applications require watermarks to convey richer information, such as user identifiers and timestamps, which zero-bit methods cannot support. Hence, another line of work focuses on embedding a multi-bit message (i.e., a binary string containing several bits). Existing solutions fall into two main categories: the first splits the vocabulary into distinct groups (like color categories) to encode data bits~\cite{fernandez2023three}, and the second divides the text into chunks with each chunk holding part of the payload~\cite{wang2023towards}. However, when either method partitions too finely (e.g., too many groups or tiny chunks), the watermark becomes weak and hard to detect. To address this, Yoo et al.~\cite{yoo2024advancing} merge both approaches, avoiding extreme splitting while maintaining enough data capacity. More recently, Qu et al.~\cite{qu2025provably} empirically observe that existing schemes suffer from either decoding efficiency or detection accuracy when encoding larger bit payloads, and propose a segmented construction to balance the two; Xu et al.~\cite{xu2024robust} further refine payload robustness through majority-bit-aware aggregation. Another direction focuses on designing robust mechanisms that can trace content in settings with many adaptive users attempting to evade detection~\cite{cohen2025watermarking}, establishing rigorous statistical guarantees for ultra-low false positive rates while extending the mechanism to multi-bit message embedding.

\subsection{Semantic-level Watermarking}
The aforementioned methods operate primarily at the token level, making watermark signals vulnerable to adversarial rewriting such as paraphrasing. Sentence-level semantic watermarking (SWM) addresses this by embedding signals into sentence semantics, thereby offering robustness against lexical rewrites.

A representative work is SemStamp~\cite{semstamp}, which extends the idea of prefix-hashing from KGW into the sentence semantic space. Specifically, it employs locality-sensitive hashing~\cite{lsh1} to compute a hash from the preceding sentence, which defines the ``green'' region for embedding the subsequent sentence. During detection, SemStamp counts the number of sentences whose embeddings fall within the corresponding green regions, conditioned on their preceding sentences, and leverages this statistic to perform a watermark-detection test. A follow-up, k-SemStamp~\cite{ksemstamp}, mitigates the quality degradation caused by SemStamp's rigid restriction to a single semantic partition by allowing sampling from the $k$ nearest-neighbor partitions, thereby improving fluency and semantic preservation. Similarly, SemaMark~\cite{ren2023robust} also relies on prefix hashing but adopts a Normalized Embedding Ring strategy. Beyond hash-based approaches, CoheMark~\cite{cohemark} and SimMark~\cite{simmark} explore more flexible watermarking criteria using fuzzy c-means clustering or sentence similarity, while PersonaMark~\cite{zhang2024personamark} hashes sentence structures for personalized embedding. PMark~\cite{huo2026pmark} addresses the distorted sentence distribution observed in prior methods and proposes a median-estimation-based approach to achieve distortion-free semantic-level watermarking with multi-channel constraints. Other recent efforts explore sparse autoencoders~\cite{saemark} and LLM-based paraphrasers~\cite{xu2024robust} for sentence-level multi-bit embedding.

However, existing SWM methods inherit KGW's \emph{prefix-based} design: the watermark signal in each sentence is pseudorandomly conditioned on its preceding context. Such dependence on sentence order makes them highly vulnerable to paragraph-level paraphrasing attacks that involve sentence reordering, splitting, or merging. Meanwhile, quality control in SWM remains underexplored, as simple n-gram filters fail to address semantic-level redundancy across sentences.

Our work falls within the domain of semantic-level watermarking but is inspired by token-level designs that couple watermarks with a global secret key sequence~\cite{kuditipudi2024robust}, a concept also explored in prior works including undetectable watermarks~\cite{christ2024undetectable}, bi-level signatures~\cite{zhou2024bileve}, and publicly-detectable watermarks~\cite{fairoze2025publicly}. These approaches either operate at the token level or fail to maintain robustness under structural or paragraph-level perturbations. Meanwhile, although both our method and existing multi-bit watermarking approaches embed multi-bit-like signals into text, our goal is not to transmit a specific payload across the entire document; instead, we propose a \emph{self-anchored} semantic-level watermarking framework that severs the reliance on contextual order and incorporates diversity-aware filtering for quality preservation, pushing the Pareto frontier of the robustness-quality trade-off in text watermarking.
\section{Limitations and Future Work}\label{sec:dis}
Despite its strong paragraph-level paraphrase robustness, \method has several limitations. First, like other semantic-level watermarking methods, it requires sampling multiple sentence candidates before selecting one that satisfies the target constraint, since current autoregressive LLMs are decoded token by token rather than sentence by sentence. This extra sampling loop makes generation latency higher than that of token-level methods. Second, our diversity-aware filtering improves text quality in practice but does not yet provide a formal distortion-free guarantee, because its hard filters and soft regularizers still modify the effective sampling distribution. Finally, our work focuses on zero-bit watermarking, which embeds only a yes/no signal into generated text and is not directly applicable to multi-bit scenarios that require payload encoding.

Several directions could address these limitations. On the efficiency side, emerging architectures such as diffusion language models and large concept models operate natively over sentence- or concept-level units and may provide a more direct interface for semantic watermarking, removing the need for repeated candidate sampling; complementary techniques such as steering generation could further guide the model toward valid candidates in a single pass. On the distortion side, a deeper theoretical understanding of the geometry of semantic embedding spaces, beyond the empirically motivated hyperbolic transform used here, could inform sampling strategies that remain robust to paragraph-level paraphrasing while achieving a formal distortion-free guarantee. For the multi-bit extension, the self-anchored partition can be generalized from a binary green/red split to a multi-way partition of the semantic space, allowing each sentence to carry several payload bits while preserving robustness to paraphrasing. We leave these as future work.

\section{Conclusion}\label{sec:con}
Current token-level and semantic-level watermarking methods rely on context-aware random seeders or predefined private keys during generation. However, this dependency leaves them highly vulnerable to paragraph-level paraphrasing attacks that involve word substitution and sentence permutation, severely compromising watermark traceability in adversarial settings. To address this, we introduce self-anchored watermarking, a concept in which the green region at timestep $t$ is determined exclusively by the semantics of the $t$-th sentence. By removing reliance on contextual order, this design improves resilience to global paraphrasing attacks. To the best of our knowledge, \method is the first watermarking framework to achieve robustness against paragraph-level paraphrasing. Experimental results show that our approach consistently surpasses existing SWM baselines in robustness while maintaining competitive generation quality, establishing a more reliable paradigm for semantic-level watermarking.


\section*{LLM Usage Considerations}
We used AI assistants for two purposes: (1) generating routine code and boilerplate functions, which were subsequently reviewed and debugged by humans, and (2) performing grammatical review and sentence-level editing of the manuscript. The research methodology, findings, and analysis were independently proposed and conducted.
\section*{Ethics Considerations}
This work studies text watermarking for identifying AI-generated content. Our experiments use publicly available datasets and generated model outputs, and do not involve recruiting human subjects or collecting personal information. Potential risks include false-positive detections and misuse of watermark detectors for unsupported attribution decisions. We therefore report low-FPR detection metrics and emphasize that watermarking evidence should be interpreted as a statistical signal rather than definitive proof of authorship.

\section*{Acknowledgment}
Mingxun Zhou is supported by the National Key Research and Development Program of China (2025YFE0200500), and a Hong Kong RGC ECS grant (26205726).

\clearpage
\bibliographystyle{IEEEtran}
\bibliography{references}

\appendices
\section{Experiment Setup Details}\label{app:setup}
\subsection{Baselines}
\textbf{MarkLLM.} We use the official MarkLLM implementation~\cite{markllm}\footnote{\url{https://github.com/THU-BPM/MarkLLM}} to reproduce the results of various token-level baselines, including KGW~\cite{kgw}, UPV~\cite{upv}, MorphMark~\cite{morphmark}, EXP~\cite{exp}, and SynthID~\cite{synthid}. During generation, we set \texttt{max\_new\_tokens} to 256, which is comparable to the generation length used by semantic-level methods. The temperature and \texttt{top\_p} are fixed at $0.7$ and $0.95$, respectively, across all experiments to ensure a fair comparison. In addition, for baselines that depend on external networks (e.g., UPV~\cite{upv}), we use the official weights provided by MarkLLM.

\textbf{$k$-SemStamp.} We adapt the official $k$-SemStamp implementation from MarkLLM to align it with other semantic-level methods. Specifically, we set \texttt{max\_new\_sentences} to 12 (instead of a token count), and apply this setting consistently across all semantic-level baselines. We set \texttt{max\_trials} to the default value of $100$ for SemStamp and $k$-SemStamp. For the embedding model, we use the fine-tuned \texttt{all-mpnet-base-v2}~\cite{reimers-2019-sentence-bert}\footnote{\url{https://huggingface.co/AbeHou/SemStamp-c4-sbert}} provided by \cite{ksemstamp}. We also use the $k$-means centroid weights\footnote{\url{https://github.com/abehou/SemStamp}} released by~\cite{ksemstamp}, trained on the C4~\cite{c4} and BookSum~\cite{booksum} datasets. All other hyperparameters are kept at their default settings in the $k$-SemStamp implementation of MarkLLM.

\textbf{SemStamp.} For SemStamp~\cite{semstamp}, we use the authors' official implementation, setting temperature $=0.7$, \texttt{top\_p} $=0.95$, and \texttt{max\_new\_sentences} $=12$ to match the other semantic-level methods. We also use the fine-tuned embedding model for a fair comparison.

\textbf{PMark.} For PMark~\cite{huo2026pmark}, we use the authors' official implementation of the online PMark variant, setting temperature $=0.7$, \texttt{top\_p} $=0.95$, and \texttt{max\_new\_sentences} $=12$ to match the other semantic-level methods, with \texttt{all-mpnet-base-v2} as the embedding model. The channel number is set to $4$, and the sample budget is $N=64$.

\textbf{\method.} During generation, we use the original \texttt{all-mpnet-base-v2} embedding model~\cite{reimers-2019-sentence-bert}\footnote{\url{https://huggingface.co/sentence-transformers/all-mpnet-base-v2}} without fine-tuning. Pivot vectors are generated via QR decomposition of a Gaussian random matrix to ensure orthogonality. We set temperature $=0.7$, \texttt{top\_p} $=0.95$, and \texttt{max\_new\_sentences} $=12$, consistent with other semantic-level methods. The sample budget is $N=64$, and the number of channels is $b=2$.

\subsection{Prompt Templates}\label{app:prompt_templates}
We summarize the prompt templates used in the evaluation pipeline.

\textbf{Doc-T.} The back-translation attack uses two translation prompts:
\begin{itemize}
    \item EN$\rightarrow$ZH: ``Translate the following text from English to Chinese. Only output the translation, nothing else:''
    \item ZH$\rightarrow$EN: ``Translate the following text from Chinese to English. Only output the translation, nothing else:''
\end{itemize}

\textbf{LLM pairwise judge.} For blind ranking, we use pairwise comparison with dataset-specific system prompts:
\begin{itemize}
    \item \texttt{booksum}: ``Compare Summary A/B by accuracy, completeness, coherence, writing quality, and absence of obvious errors.''
    \item \texttt{c4}: ``Compare Continuation A/B by topical relevance, coherence, writing quality, informativeness, and absence of obvious errors.''
\end{itemize}
Both settings require JSON-only output, as shown in Figure~\ref{fig:pairwise_output_template}.
\begin{figure}[h]
    \centering
    \vspace{-2mm}
    \includegraphics[width=\linewidth]{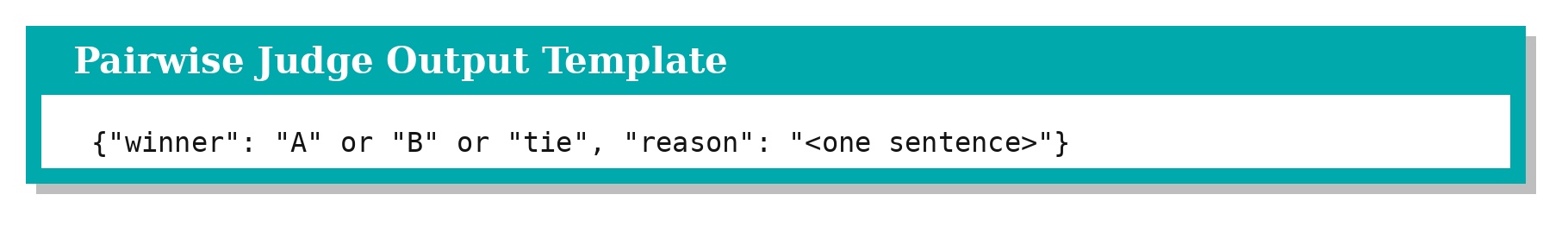}
    \caption{Output template for the LLM pairwise judge.}
    \label{fig:pairwise_output_template}
    \vspace{-2mm}
\end{figure}

For example, the pairwise user query template of \texttt{booksum} is shown in Figure~\ref{fig:booksum_pairwise_query_template}.
\begin{figure}[h]
    \centering
    \vspace{-2mm}
    \includegraphics[width=\linewidth]{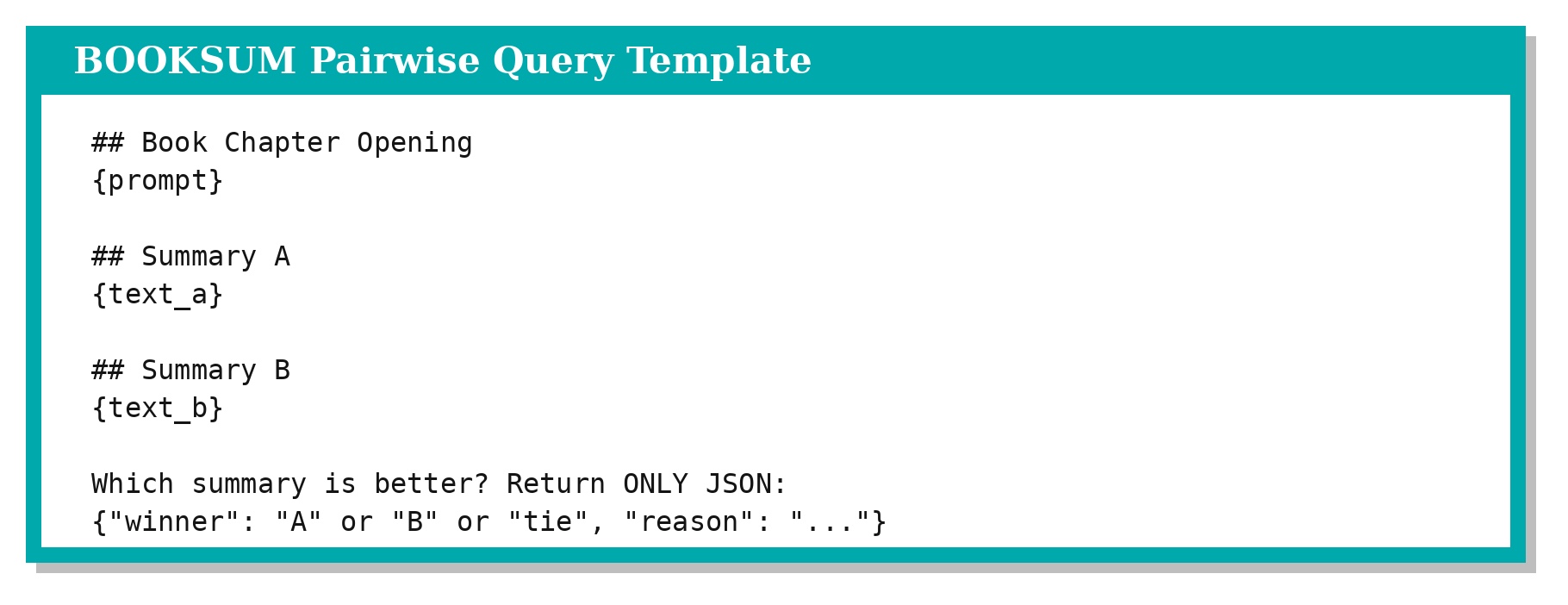}
    \caption{Pairwise user query template for \texttt{booksum}.}
    \label{fig:booksum_pairwise_query_template}
    \vspace{-2mm}
\end{figure}

To ensure blind evaluation, watermarked and unwatermarked texts are assigned to A/B in randomized order for each sample, and final win/lose/tie is mapped back to the watermarked side.

\textbf{LLM ranking prompts.} For multi-method blind ranking, we evaluate multiple candidates (A/B/C/...) jointly with dataset-specific system prompts:
\begin{itemize}
    \item \texttt{booksum}: ``Rank all summaries by accuracy, completeness, coherence, writing quality, and absence of obvious errors.''
    \item \texttt{c4}: ``Rank all continuations by topical relevance, coherence, writing quality, informativeness, and absence of obvious errors.''
\end{itemize}
The output is constrained to JSON, as shown in Figure~\ref{fig:ranking_output_template}, where the first element is best. Ties are explicitly allowed using entries such as \texttt{"A=B"}.
\begin{figure}[h]
    \centering
    \vspace{-2mm}
    \includegraphics[width=\linewidth]{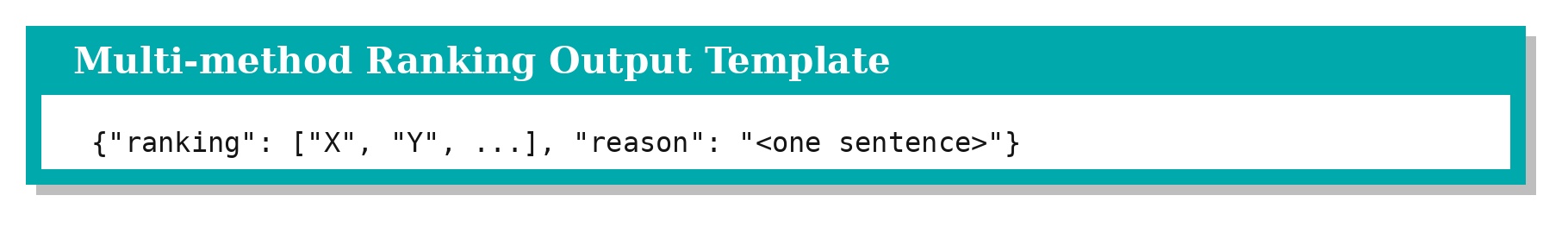}
    \caption{Output template for multi-method blind ranking.}
    \label{fig:ranking_output_template}
    \vspace{-2mm}
\end{figure}

For instance, the ranking query template of \texttt{booksum} is shown in Figure~\ref{fig:booksum_ranking_query_template}.
\begin{figure}[h]
    \centering
    \vspace{-2mm}
    \includegraphics[width=\linewidth]{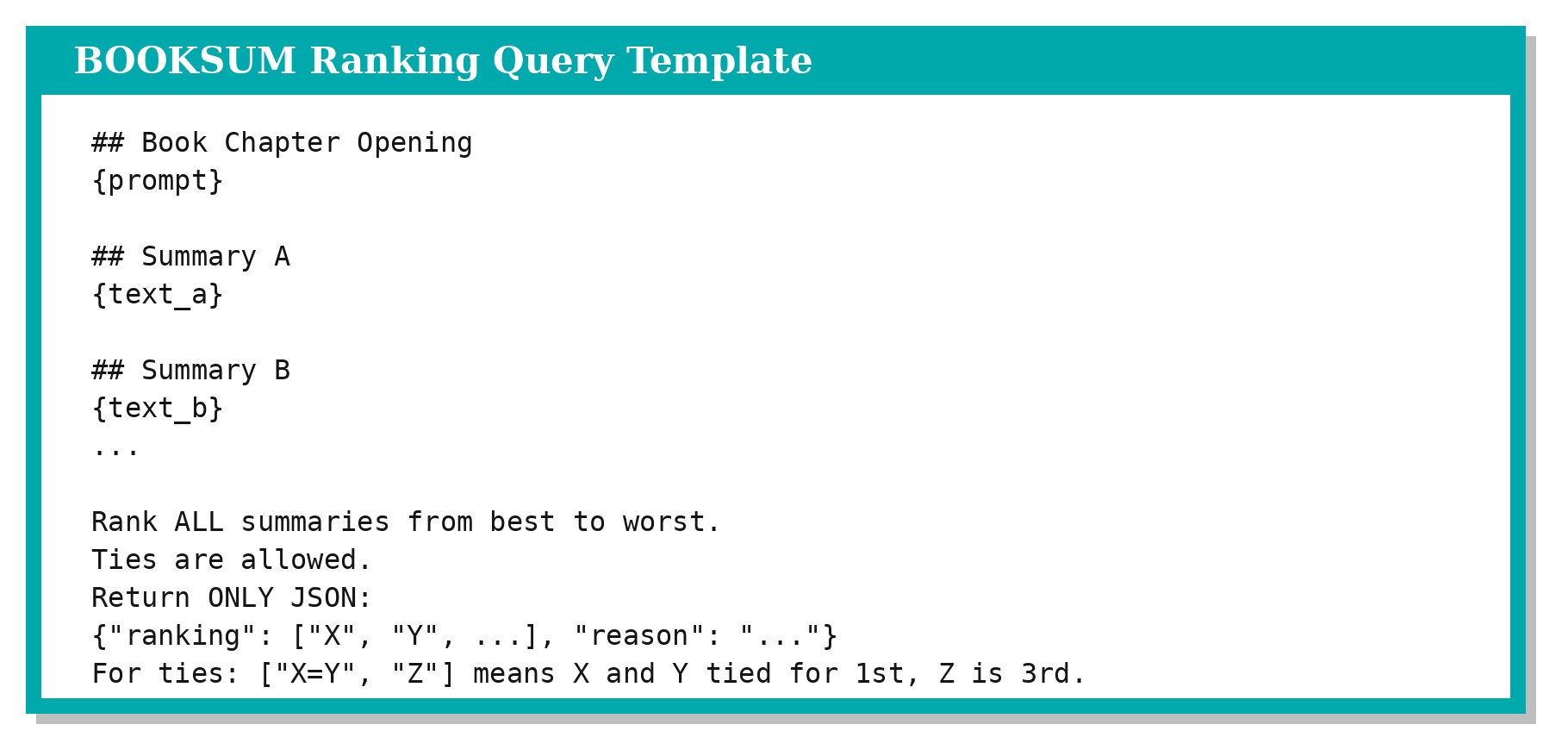}
    \caption{Ranking user query template for \texttt{booksum}.}
    \label{fig:booksum_ranking_query_template}
    \vspace{-2mm}
\end{figure}

For each sample, method outputs are first randomly mapped to blind labels A/B/C/..., then mapped back after parsing.

\subsection{Quality Metrics}\label{app:quality_metrics}

For each generated sample $x$, we compute three automatic diversity metrics in addition to the pairwise LLM judge. Let $\mathcal{S}(x)$ denote the sequence of normalized sentences in $x$ after sentence tokenization, trimming, and lowercasing, and let $\mathcal{G}_n(x)$ denote the sequence of word $n$-grams obtained after lowercasing and whitespace tokenization. Sentence Duplicate percentage (SD) is the fraction of generated sentences that duplicate an earlier sentence in the same output:
\begin{equation}
\operatorname{SD}(x) =
100 \times
\frac{|\mathcal{S}(x)| - |\operatorname{uniq}(\mathcal{S}(x))|}{|\mathcal{S}(x)|}.
\end{equation}
Distinct-2 (D-2) measures bigram diversity:
\begin{equation}
\operatorname{D\text{-}2}(x) =
\frac{|\operatorname{uniq}(\mathcal{G}_2(x))|}{|\mathcal{G}_2(x)|}.
\end{equation}
The 4-gram Repeat percentage (4g) measures repetitions beyond the first occurrence of each 4-gram:
\begin{equation}
\operatorname{4g}(x) =
100 \times
\frac{|\mathcal{G}_4(x)| - |\operatorname{uniq}(\mathcal{G}_4(x))|}{|\mathcal{G}_4(x)|}.
\end{equation}
If the denominator is zero, the corresponding metric is set to $0$. We report the average of each metric over the evaluation set. Lower SD and 4g indicate less repetition, while higher D-2 indicates greater lexical diversity.

\subsection{Attack Details}\label{app:attack}
\begin{figure}[t]
\centering
\includegraphics[width=\linewidth]{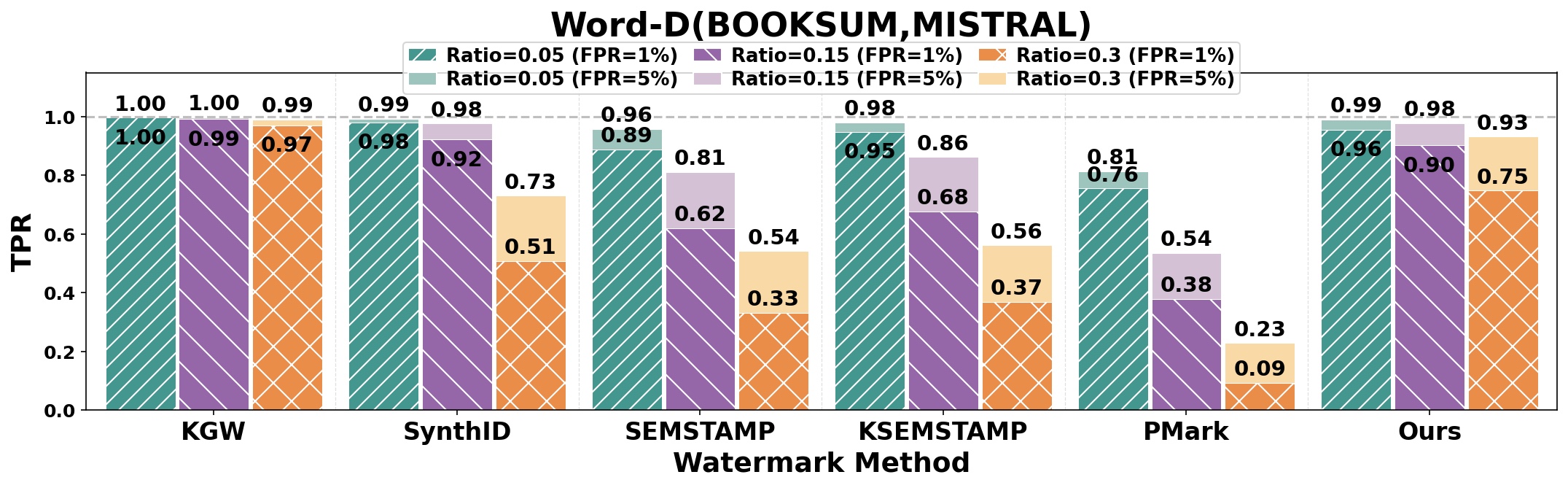}
\vspace{-2mm}
\caption{Word-D attack results.}
\label{fig:word_d}
\end{figure}
\textbf{Doc-P.} We conduct document-level paraphrase attacks using GPT-4.1-mini~\cite{gpt4}. Two prompt templates are used to instruct paraphrasing at different adversarial strengths. Attack Prompt I: ``Please rewrite the following text:''; Attack Prompt II: ``Please rewrite the following text, avoiding using the same words or phrases in the original text:''. All attacks are applied at the paragraph level across methods and settings to simulate real-world paragraph-level paraphrasing scenarios.

\textbf{Doc-T.} We implement the back-translation attack following MarkLLM~\cite{markllm}. Specifically, we use GPT-4.1-mini to translate watermarked English text into Chinese and then back into English in an attempt to remove watermark evidence.

\textbf{Word-D and Word-S.} We implement word deletion (Word-D) and synonym substitution (Word-S) attacks following the official MarkLLM implementation. We set the attack ratios to 5\%, 15\%, and 30\% to evaluate the overall robustness of our method. We apply word-level attacks to the whole paragraph rather than sentence by sentence, unlike PMark~\cite{huo2026pmark}. Additional experimental results are reported in Appendix~\ref{app:word_attack}.

\section{Additional Experiments}\label{app:full_results}

\subsection{Full Results}
Tables~\ref{tab:booksum_full} and~\ref{tab:c4_full} present the complete experimental results with all metrics.

\begin{table*}[t!]
\centering
\caption{Full results on \textbf{BOOKSUM} dataset. For detection, we report TP@FP=1\%/TP@FP=5\%/AUC (\%). \textbf{Doc-P} denotes document-level paraphrase attack (GPT) with Attack Prompt I and II; \textbf{Doc-T} denotes document-level translation attack (GPT). For text quality, we report Sentence Duplicate\% (SD, lower is better), Distinct-2 (D-2, higher is better), and pairwise LLM Judge Win/Lose/Tie\% against unwatermarked text. \textbf{Bold} = best, \underline{underlined} = second-best.}
\vspace{-2mm}
\begin{adjustbox}{width=\textwidth,center}
\begin{tabular}{l|cccc|ccccc}
\toprule
\multirow{2}{*}{\textbf{Method}}
& \multicolumn{4}{c|}{\textit{Detection}}
& \multicolumn{5}{c}{\textit{Text Quality}} \\
\cline{2-10}
& \textbf{No Atk} & \textbf{Doc-P I} & \textbf{Doc-P II} & \textbf{Doc-T(GPT)}
& \textbf{SD\%}$\downarrow$ & \textbf{D-2}$\uparrow$ & \textbf{PS-W\%} & \textbf{PS-L\%} & \textbf{PS-T\%}\\
\midrule
\multicolumn{10}{c}{\textbf{Mistral-Small-3.1-24B-Base-2503}} \\
\midrule
SemStamp~\cite{semstamp} & 96.0/98.0/99.5 & 49.8/64.3/91.1 & 39.8/54.6/85.7 & 58.6/72.7/92.1 & 8.9 & 0.761 & 53.2 & 39.2 & 7.6 \\
$k$-SemStamp~\cite{ksemstamp} & 98.6/99.6/99.9 & 61.7/78.8/94.1 & 43.1/62.1/88.3 & 65.9/81.6/\underline{95.3} & 9.5 & 0.756 & 51.4 & 47.2 & 1.4 \\
KGW~\cite{kgw} & 100.0/100.0/100.0 & \underline{64.0}/\underline{84.0}/\underline{96.6} & 34.4/59.6/89.7 & 54.4/76.0/95.0 & 0.4 & 0.948 & 38.0 & 50.2 & 11.8 \\
SynthID~\cite{synthid} & 99.8/99.8/100.0 & 23.8/47.2/80.8 & 10.4/24.2/64.5 & 12.6/35.4/75.1 & 0.2 & 0.910 & 54.2 & 38.6 & 7.2 \\
MorphMark~\cite{morphmark} & 99.4/99.8/100.0 & 56.0/78.6/94.7 & 34.2/55.8/86.9 & 45.8/74.0/92.7 & 0.2 & 0.942 & 47.6 & 51.4 & 1.0 \\
UPV~\cite{upv} & 99.0/99.4/98.7 & 61.4/77.0/85.7 & 23.6/39.6/65.9 & 38.0/57.2/75.3 & 0.7 & 0.928 & 47.2 & 51.4 & 1.4 \\
EXP~\cite{exp} & 99.8/99.8/100.0 & 38.2/57.4/86.8 & 17.2/25.0/69.9 & 29.6/45.2/82.8 & 10.2 & 0.770 & 52.4 & 46.2 & 1.4 \\
PMark~\cite{huo2026pmark} & 98.0/98.4/99.4 & 51.5/66.7/89.9 & \underline{57.5}/\underline{71.5}/\underline{92.0} & \underline{77.0}/\underline{82.8}/94.3 & 4.0 & 0.788 & 57.2 & 37.6 & 5.2 \\
\method (Ours) & 98.0/99.6/99.8 & \textbf{90.5}/\textbf{97.9}/\textbf{99.5} & \textbf{88.1}/\textbf{97.4}/\textbf{99.3} & \textbf{87.1}/\textbf{97.1}/\textbf{99.1} & 0.1 & 0.816 & 47.9 & 45.2 & 6.9 \\
\midrule
\multicolumn{10}{c}{\textbf{NVIDIA-Nemotron-Nano-9B-v2-Base}} \\
\midrule
SemStamp~\cite{semstamp} & 97.0/98.2/99.7 & 46.9/63.9/90.2 & 34.5/50.5/86.3 & 64.3/76.6/93.9 & 15.9 & 0.687 & 39.2 & 41.6 & 19.2 \\
$k$-SemStamp~\cite{ksemstamp} & 97.6/98.8/99.2 & 51.7/74.3/93.4 & 37.1/58.9/88.6 & 63.7/79.0/93.2 & 17.9 & 0.696 & 35.0 & 45.0 & 20.0 \\
KGW~\cite{kgw} & 100.0/100.0/100.0 & \underline{76.8}/\underline{90.0}/\underline{97.8} & 45.6/\underline{72.2}/\underline{93.1} & \underline{66.8}/85.2/\underline{97.0} & 0.1 & 0.941 & 32.0 & 42.4 & 25.6 \\
SynthID~\cite{synthid} & 99.8/100.0/100.0 & 31.4/52.6/82.8 & 15.6/29.8/68.4 & 21.2/46.0/81.6 & 0.1 & 0.897 & 47.8 & 33.4 & 18.8 \\
MorphMark~\cite{morphmark} & 100.0/100.0/100.0 & 74.0/88.4/97.7 & 46.2/70.8/92.7 & 64.0/\underline{85.6}/96.5 & 0.2 & 0.937 & 36.4 & 40.4 & 23.2 \\
UPV~\cite{upv} & 99.4/99.6/98.8 & 71.0/83.4/87.7 & \underline{48.8}/65.8/78.9 & 61.6/77.6/85.4 & 0.4 & 0.916 & 39.2 & 38.2 & 22.6 \\
EXP~\cite{exp} & 99.8/99.8/100.0 & 42.2/57.6/86.9 & 16.2/25.4/71.7 & 37.8/55.0/83.5 & 16.7 & 0.686 & 30.6 & 45.8 & 23.6 \\
PMark~\cite{huo2026pmark} & 96.6/98.4/99.5 & 24.2/43.5/82.0 & 28.3/48.1/84.0 & 46.3/60.1/87.9 & 15.1 & 0.732 & 26.0 & 49.8 & 24.2 \\
\method (Ours) & 98.8/99.8/99.9 & \textbf{92.6}/\textbf{97.6}/\textbf{99.4} & \textbf{90.2}/\textbf{96.4}/\textbf{99.1} & \textbf{88.2}/\textbf{96.2}/\textbf{99.0} & 0.6 & 0.746 & 34.2 & 46.0 & 19.8 \\
\bottomrule
\end{tabular}
\end{adjustbox}
\label{tab:booksum_full}
\end{table*}

\begin{table*}[t!]
\centering
\caption{Full results on \textbf{C4} dataset. For detection, we report TP@FP=1\%/TP@FP=5\%/AUC (\%). \textbf{Doc-P} denotes document-level paraphrase attack (GPT) with Attack Prompt I and II; \textbf{Doc-T} denotes document-level translation attack (GPT). For text quality, we report Sentence Duplicate\% (SD, lower is better), Distinct-2 (D-2, higher is better), and pairwise LLM Judge Win/Lose/Tie\% against unwatermarked text. \textbf{Bold} = best, \underline{underlined} = second-best.}
\vspace{-2mm}
\begin{adjustbox}{width=\textwidth,center}
\begin{tabular}{l|cccc|ccccc}
\toprule
\multirow{2}{*}{\textbf{Method}}
& \multicolumn{4}{c|}{\textit{Detection}}
& \multicolumn{5}{c}{\textit{Text Quality}} \\
\cline{2-10}
& \textbf{No Atk} & \textbf{Doc-P I} & \textbf{Doc-P II} & \textbf{Doc-T(GPT)}
& \textbf{SD\%}$\downarrow$ & \textbf{D-2}$\uparrow$ & \textbf{PS-W\%} & \textbf{PS-L\%} & \textbf{PS-T\%}\\
\midrule
\multicolumn{10}{c}{\textbf{Mistral-Small-3.1-24B-Base-2503}} \\
\midrule
SemStamp~\cite{semstamp} & 95.1/99.0/99.7 & 28.6/58.1/87.4 & 17.2/31.0/73.2 & 54.7/77.8/92.1 & 8.4 & 0.788 & 56.4 & 43.6 & 0.0 \\
$k$-SemStamp~\cite{ksemstamp} & 92.5/99.5/99.0 & 22.5/59.0/87.5 & 15.5/40.2/79.8 & 44.5/74.5/91.1 & 5.3 & 0.821 & 58.4 & 41.4 & 0.2 \\
KGW~\cite{kgw} & 100.0/100.0/100.0 & \underline{77.7}/\textbf{90.7}/\textbf{97.6} & 29.0/47.2/88.3 & 64.2/84.5/\underline{96.3} & 0.0 & 0.953 & 41.6 & 57.8 & 0.6 \\
SynthID~\cite{synthid} & 99.2/99.2/99.7 & 44.1/58.3/84.4 & 25.1/37.2/69.9 & 44.1/59.4/83.6 & 0.1 & 0.921 & 59.1 & 40.5 & 0.4 \\
MorphMark~\cite{morphmark} & 99.7/100.0/100.0 & 55.5/77.4/95.4 & 28.0/51.4/\underline{89.2} & 52.7/76.8/95.8 & 0.0 & 0.947 & 41.7 & 57.9 & 0.4 \\
UPV~\cite{upv} & 99.7/100.0/99.7 & 67.1/85.9/84.3 & 44.5/\underline{74.6}/75.2 & \underline{71.5}/\underline{92.8}/85.8 & 0.1 & 0.931 & 44.0 & 55.6 & 0.4 \\
EXP~\cite{exp} & 99.4/99.4/99.4 & 24.6/48.0/82.6 & 11.5/22.7/70.5 & 36.4/60.4/88.0 & 6.6 & 0.810 & 50.8 & 49.0 & 0.2 \\
PMark~\cite{huo2026pmark} & 96.2/96.4/98.9 & 40.0/54.0/84.3 & \underline{48.9}/60.4/86.4 & 67.6/73.8/92.8 & 3.9 & 0.811 & 64.8 & 35.0 & 0.2 \\
\method (Ours) & 94.5/97.6/99.5 & \textbf{78.6}/\underline{89.3}/\underline{97.5} & \textbf{76.8}/\textbf{87.5}/\textbf{97.0} & \textbf{84.2}/\textbf{93.0}/\textbf{98.3} & 0.5 & 0.814 & 82.0 & 5.0 & 13.0 \\
\midrule
\multicolumn{10}{c}{\textbf{NVIDIA-Nemotron-Nano-9B-v2-Base}} \\
\midrule
SemStamp~\cite{semstamp} & 96.4/98.0/99.0 & 40.4/54.8/87.4 & 27.9/44.2/78.4 & 61.7/71.6/90.6 & 14.3 & 0.700 & 52.4 & 44.8 & 2.8 \\
$k$-SemStamp~\cite{ksemstamp} & 90.5/99.2/98.2 & 15.6/49.2/84.4 & 14.6/40.3/81.1 & 32.1/65.7/88.4 & 13.9 & 0.721 & 39.2 & 56.2 & 4.6 \\
KGW~\cite{kgw} & 100.0/100.0/100.0 & 55.9/80.3/\underline{96.0} & 27.3/49.9/\underline{88.6} & 56.2/80.6/\underline{96.2} & 0.2 & 0.957 & 39.6 & 56.0 & 4.4 \\
SynthID~\cite{synthid} & 100.0/100.0/100.0 & 42.9/59.9/84.1 & 22.6/35.0/69.0 & 53.0/69.3/87.6 & 0.0 & 0.904 & 54.8 & 41.4 & 3.8 \\
MorphMark~\cite{morphmark} & 99.8/100.0/100.0 & 48.3/71.4/93.8 & 20.4/41.5/85.8 & 49.0/78.1/95.3 & 0.1 & 0.951 & 45.0 & 50.8 & 4.2 \\
UPV~\cite{upv} & 99.7/100.0/99.8 & \underline{75.0}/\textbf{91.7}/87.9 & \underline{53.2}/\underline{79.2}/77.7 & \underline{79.5}/\textbf{94.6}/88.0 & 0.2 & 0.923 & 42.2 & 53.2 & 4.6 \\
EXP~\cite{exp} & 100.0/100.0/99.8 & 39.1/50.3/83.4 & 16.9/25.1/69.7 & 58.6/70.6/90.2 & 12.4 & 0.748 & 36.7 & 58.4 & 5.0 \\
PMark~\cite{huo2026pmark} & 96.4/97.8/98.8 & 24.7/39.7/78.9 & 28.3/40.0/79.3 & 45.0/57.4/86.1 & 9.6 & 0.728 & 40.0 & 55.0 & 5.0 \\
\method (Ours) & 93.8/97.0/99.2 & \textbf{78.3}/\underline{90.6}/\textbf{97.2} & \textbf{72.9}/\textbf{89.2}/\textbf{97.1} & \textbf{82.3}/\underline{93.8}/\textbf{97.8} & 0.4 & 0.826 & 69.4 & 5.2 & 25.4 \\
\bottomrule
\end{tabular}
\end{adjustbox}
\label{tab:c4_full}
\end{table*}

\subsection{Additional Robustness--Quality Trade-off Plots}\label{app:tradeoff_more}
In Figure~\ref{fig:trade-off-appendix}, we display the robustness-quality trade-off plots for various backbones on BOOKSUM and C4.
\begin{figure*}[t]
\centering
\begin{minipage}{0.32\linewidth}
\centering
\includegraphics[width=\linewidth]{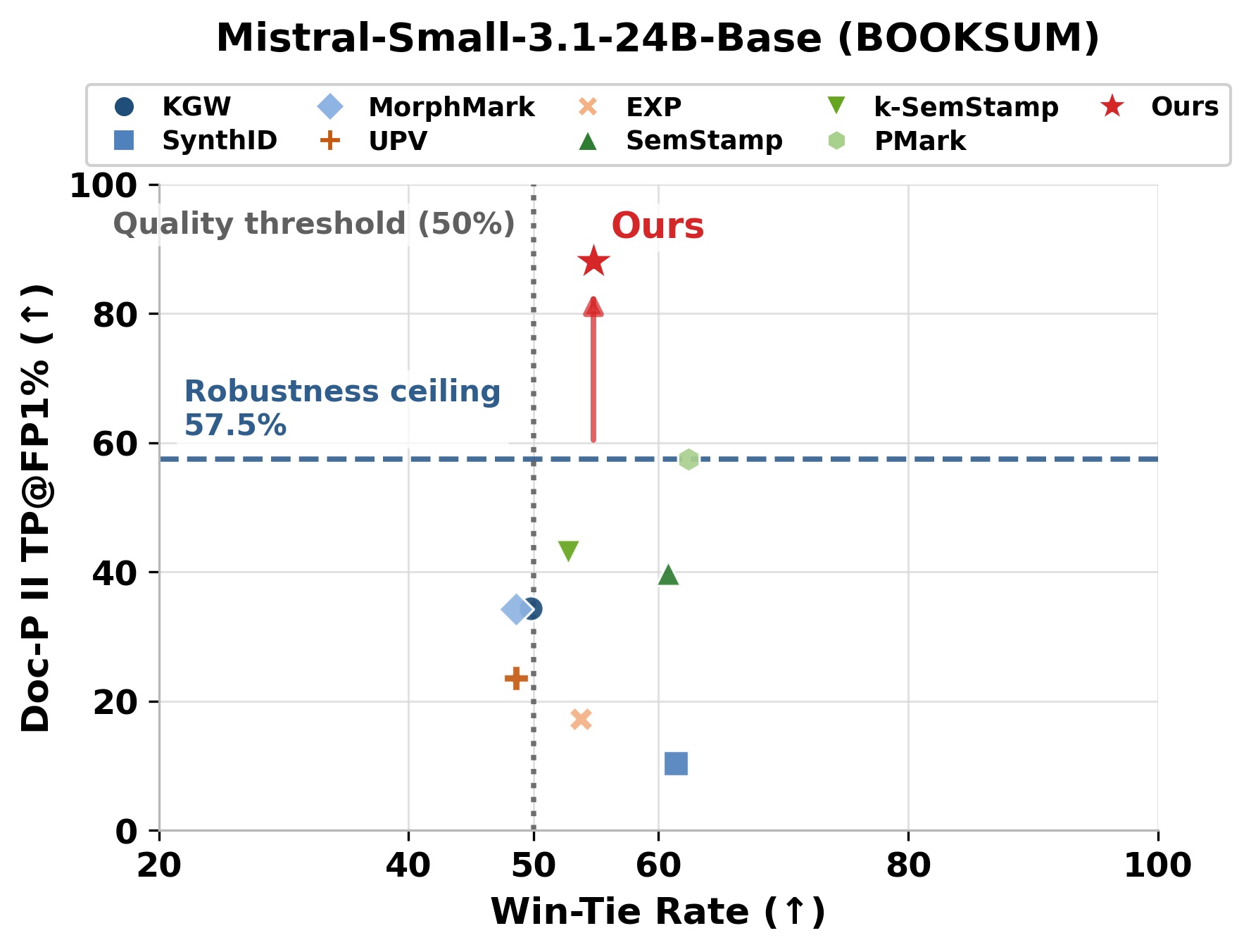}
\caption{Mistral + BOOKSUM}
\end{minipage}\hfill
\begin{minipage}{0.32\linewidth}
\centering
\includegraphics[width=\linewidth]{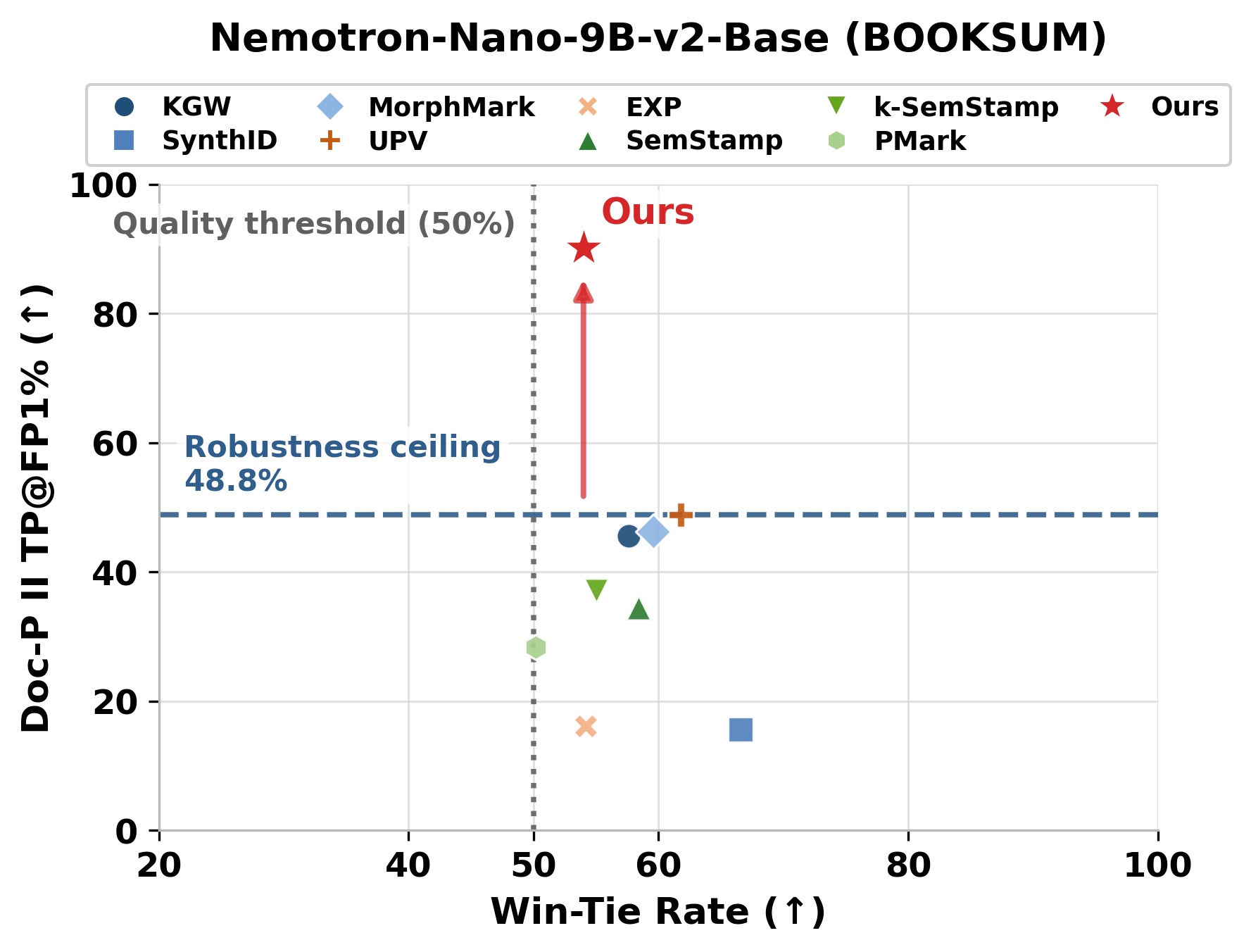}
\caption{Nemotron + BOOKSUM}
\end{minipage}\hfill
\begin{minipage}{0.32\linewidth}
\centering
\includegraphics[width=\linewidth]{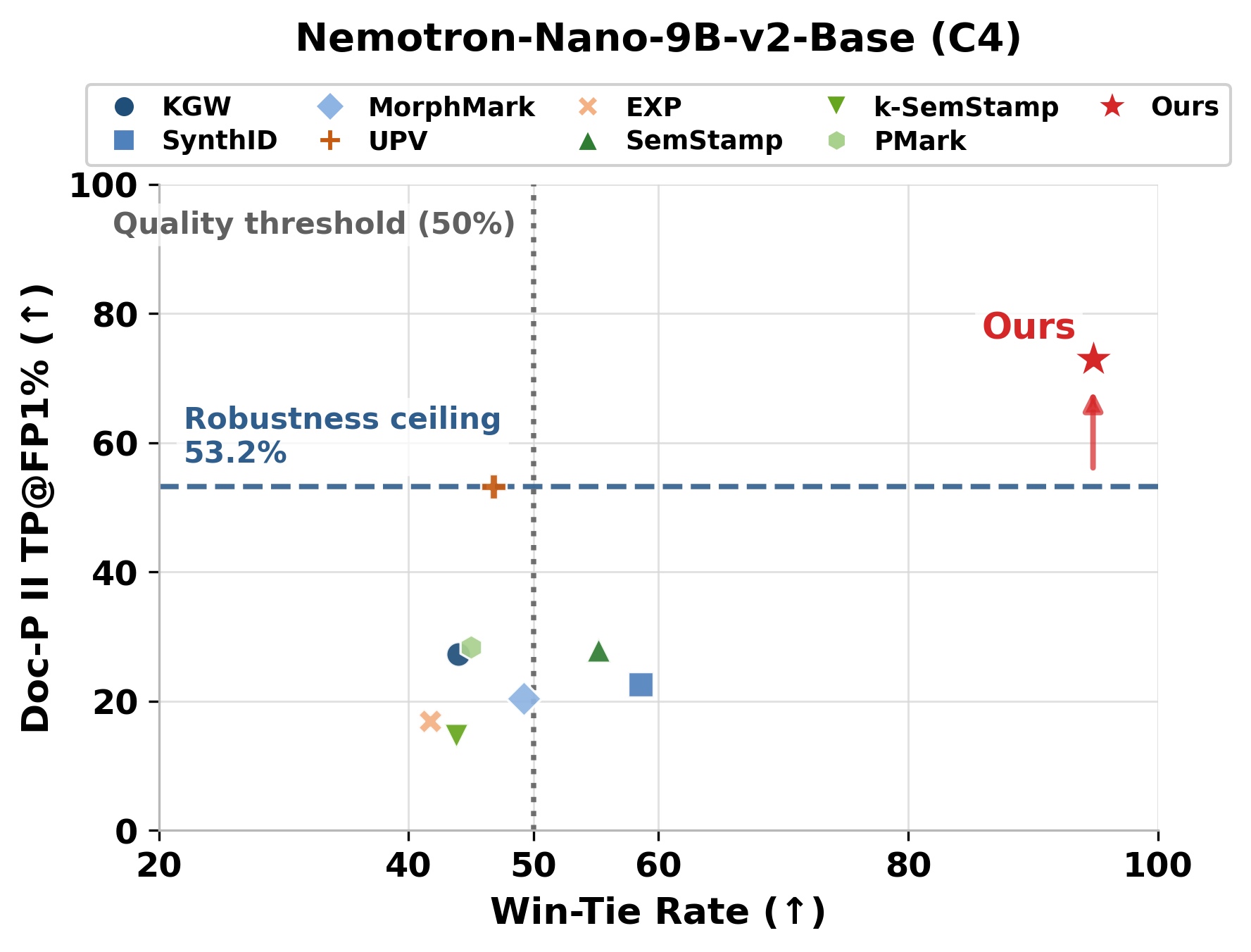}
\caption{Nemotron + C4}
\end{minipage}
\caption{Additional robustness--quality trade-off plots corresponding to the three settings not shown in the main text.}
\label{fig:trade-off-appendix}
\end{figure*}

\subsection{Word-Deletion Attack Results}\label{app:word_attack}
In Figure~\ref{fig:word_d}, we illustrate the results of various watermark methods under Word-D attacks.

\subsection{Illustration of SAMark Sampling Process}
\begin{figure*}[t] 
\centering
\includegraphics[width=0.85\linewidth]{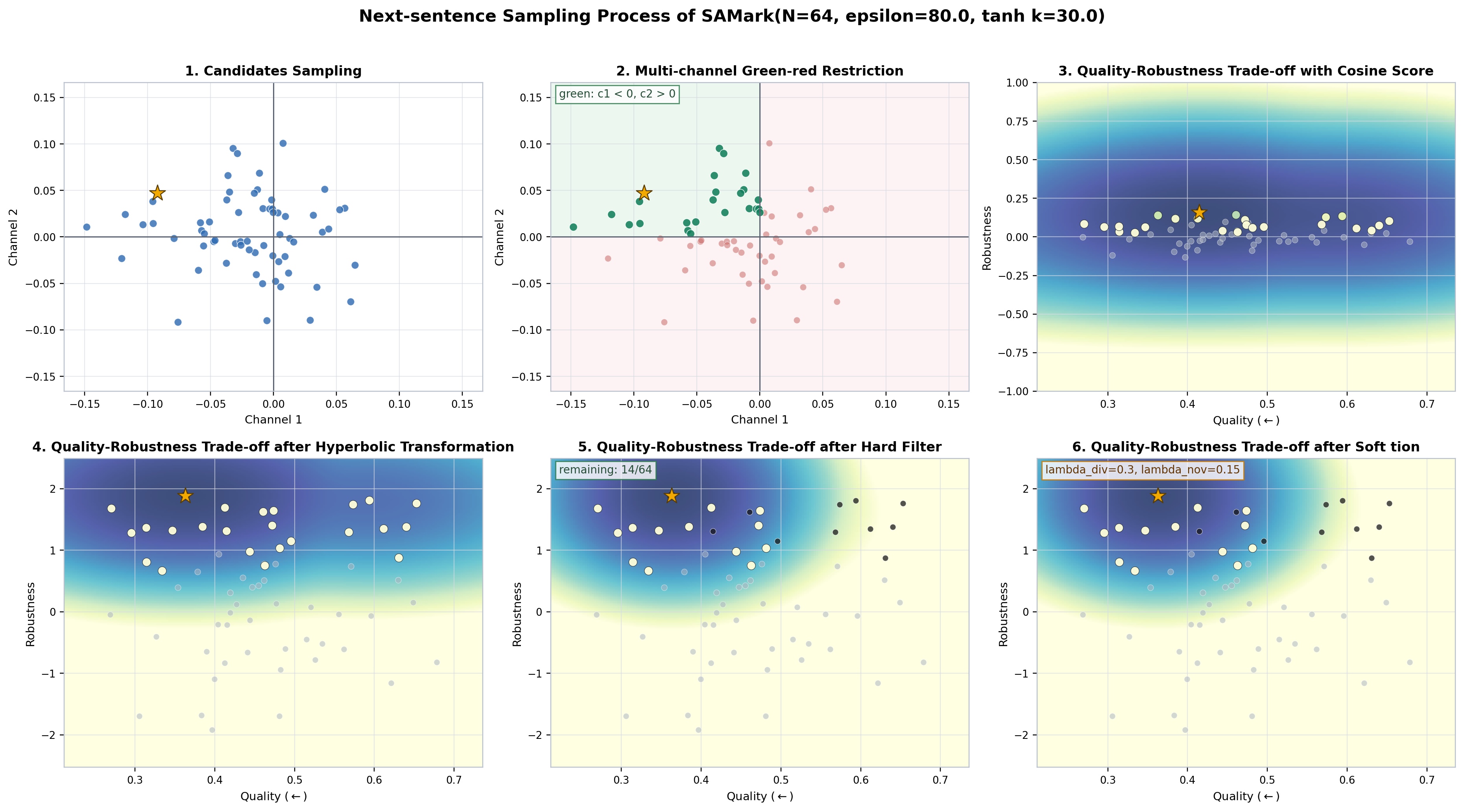}
\caption{Illustration of the next-sentence sampling process of SAMark.}
\label{fig:samark_demo} 
\end{figure*}
In Figure~\ref{fig:samark_demo}, each point denotes one sampled next-sentence candidate $s_i$ from a pool of size $N$.
Let $e_i$ be its sentence embedding, $e_{\mathrm{prev}}$ be the previous selected sentence embedding, and $\{v_k\}_{k=1}^{b}$ be channel pivots with channel signs $w_k\in\{-1,+1\}$.
We define channel cosine values
\begin{equation}
c_{ik}=\cos(e_i,v_k),\quad \tilde{c}_{ik}=w_k\,c_{ik},
\end{equation}
and the quality score
\begin{equation}
q_i=\cos(e_i,e_{\mathrm{prev}}).
\end{equation}
The robustness score before and after hyperbolic transformation is
\begin{equation}
r_i^{\mathrm{raw}}=\sum_{k=1}^{b}\tilde{c}_{ik},\qquad
r_i^{\tanh}=\sum_{k=1}^{b}\tanh(\kappa\,\tilde{c}_{ik}),
\end{equation}
where $\kappa$ is the tanh scale.

Panel 1 plots $(c_{i1},c_{i2})$ for all candidates.
Panel 2 applies the multi-channel green-red restriction and keeps candidates with matched channel signs (mask $m_i^{\mathrm{green}}$).
Panels 3--4 visualize the quality-robustness trade-off $(q_i,r_i)$ with sampling probabilities
\begin{equation}
p_i \propto \exp(\epsilon r_i),
\end{equation}
computed on active candidates under $m_i^{\mathrm{green}}$, using $r_i^{\mathrm{raw}}$ in Panel 3 and $r_i^{\tanh}$ in Panel 4.

Panel 5 adds a hard filter mask $m_i^{\mathrm{hard}}$ (e.g., overlap and semantic-similarity constraints), and samples by
\begin{equation}
p_i^{\mathrm{hard}} \propto \exp\!\big(\epsilon r_i^{\tanh}\big)\,\mathbf{1}[m_i^{\mathrm{hard}}].
\end{equation}
Panel 6 adds soft regularization with diversity and novelty:
\begin{equation}
\begin{aligned}
r_i^{\mathrm{reg}}=r_i^{\tanh}+\lambda_{\mathrm{div}}d_i+\lambda_{\mathrm{nov}}n_i,\\
p_i^{\mathrm{reg}} \propto \exp\!\big(\epsilon r_i^{\mathrm{reg}}\big)\,\mathbf{1}[m_i^{\mathrm{hard}}],
\end{aligned}
\end{equation}
where $d_i=1-\mathrm{sim}_i$ and $n_i$ is a lexical novelty bonus.
The final selected sentence is the argmax under $p_i^{\mathrm{reg}}$.

\end{document}